\documentclass[a4paper,UKenglish,cleveref, autoref, thm-restate]{lipics-v2021}

\usepackage{amsthm}
\usepackage{amsfonts} 
\usepackage{thmtools}
\usepackage{hyperref} 

\usepackage[ruled,vlined]{algorithm2e} 

\usepackage{subcaption}
\usepackage{enumitem} 
\usepackage{pgfplots} 
\pgfplotsset{compat=1.6}

\newtheorem{problem}{Problem} 

\newcommand{\ceil}[1]{\lceil #1 \rceil}
\newcommand{\floor}[1]{\lfloor #1 \rfloor}
\newcommand{\expect}[1]{\mathbb{E}[#1]}

\newcommand{\cost}[1]{\ensuremath{cost(#1)}}
\newcommand{\costFull}[2]{\ensuremath{cost(#1 ; #2)}}

\newcommand{\algProbMatchings}{\ensuremath{A_1}}
\newcommand{\algProbFixedEdges}{\ensuremath{A_3}}
\newcommand{\optProbMatchings}{\ensuremath{OPT_1}}
\newcommand{\optProbFixedEdges}{\ensuremath{OPT_3}}

\newcommand{\etal}{{\em et al.}}

\definecolor{myblue}{RGB}{80,80,160}
\definecolor{mygreen}{RGB}{80,160,80}

\title{Caching Connections in Matchings}

\author{Yaniv {Sadeh}}{Tel Aviv University, Israel}{yanivsadeh@mail.tau.ac.il}{https://orcid.org/0000-0002-5712-1028}{} 
\author{Haim {Kaplan}}{Tel Aviv University, Israel}{haimk@tau.ac.il}{https://orcid.org/0000-0001-9586-8002}{}

\authorrunning{Y. Sadeh and H. Kaplan} 

\Copyright{Yaniv Sadeh and Haim Kaplan} 

\ccsdesc[500]{Theory of computation~Caching and paging algorithms}
\ccsdesc[300]{Mathematics of computing~Graph coloring}

\keywords{Caching, Matchings, Caching in Matchings, Edge Coloring, Online Algorithms} 

\funding{The work of the authors is partially supported by Israel Science Foundation (ISF) grant number 1595-19, German Science Foundation (GIF) grant number 1367 and the Blavatnik research fund at Tel Aviv University.} 

\nolinenumbers 

\ArticleNo{0}

\begin{document}

\maketitle

\begin{abstract}
Motivated by the desire to utilize a limited number of configurable optical switches by recent advances in Software Defined Networks (SDNs), we define an online problem which we call the \emph{Caching in Matchings} problem. This problem has a natural combinatorial structure and therefore may find additional applications in theory and practice.

In the \emph{Caching in Matchings} problem our cache consists of $k$ matchings of connections between servers that form a bipartite graph. To cache a connection we insert it into one of the $k$ matchings possibly evicting at most two other connections from this matching. This problem resembles the problem known as \emph{Connection Caching}~\cite{ConnectionCaching}, where we also cache connections but our only restriction is that they form a graph with bounded degree $k$. Our results show a somewhat surprising qualitative separation between the problems: The competitive ratio of any online algorithm for caching in matchings must depend on the size of the graph.

Specifically, we give a deterministic $O(nk)$ competitive and randomized $O(n \log k)$ competitive algorithms for caching in matchings, where $n$ is the number of servers and $k$ is the number of matchings. We also show that the competitive ratio of any deterministic algorithm is $\Omega(\max(\frac{n}{k},k))$ and of any randomized algorithm is $\Omega(\log \frac{n}{k^2 \log k} \cdot \log k)$. In particular, the lower bound for randomized algorithms is $\Omega(\log n)$ regardless of $k$, and can be as high as $\Omega(\log^2 n)$ if $k=n^{1/3}$, for example. We also show that if we allow the algorithm to use at least $2k-1$ matchings compared to $k$ used by the optimum then we match the competitive ratios of connection catching which are independent of $n$. Interestingly, we also show that even a single extra matching for the algorithm allows to get substantially better bounds.
\end{abstract}


\section{Introduction}
\label{section_indtroduction}

We define the \emph{Caching in Matchings} online problem, on a fixed set of $n$ nodes. Requests are edges between these nodes. The algorithm maintains a cache of $k$ matchings, i.e.\ a $k$-edge-colorable graph. To serve a request for an edge $(u,v)$ which is not in its cache (i.e.\ a miss), the algorithm has to insert it into one of its matchings. To do this it may need to evict the edges incident to $u$ and $v$ in this specific matching. Note that an evicted edge may later be re-inserted into a different matching. The algorithm has to choose which matching to use for each miss in order to minimize its total number of misses.

One can look at this problem as a new variation of the online Connection Caching problem. In  \emph{Connection Caching}~\cite{ConnectionCaching}
the setup is the same, but the cache maintained by the algorithm must be  a graph in which each node is of degree at most $k$. In case of a miss on an edge $(u,v)$ we may choose any edge incident to $u$ and any edge incident to $v$  to evict. We do not have to maintain the edges partitioned into a particular set of $k$ matchings. Thus in Caching in Matchings we are less flexible in our eviction decisions. Once we color the new edge then the two edges we have to evict are determined.

At a first glance, the two caching problems seem  similar. In fact, the only difference is the added restriction of the coloring (matchings) that affects how the cache is maintained. Interestingly, it turns out that this seemingly small difference makes Caching in Matchings a much harder online problem compared to Connection Caching. 

\looseness=-1 
A common measure to evaluate online algorithms is their competitive ratio. We say that an online algorithm is $c$-competitive if its cost (in our case, miss count) on every input sequence is at most $c$ times the minimal possible cost for serving this sequence. One would like to design algorithms with as small $c$ as possible. The problem of Connection Caching is known to be $\Theta(k)$ (deterministic) and $\Theta(\log k)$ (randomized) competitive, and in contrast we show that the dependence on $n$ (the number of nodes) in Caching in Matchings cannot be avoided.

The motivation for our Caching in Matchings problem comes from a data-center architecture
described in \cite{BMatchingPhysicalDetails}. In this setting we have $n$ servers connected via a  communication network which is equipped with a set $O$ of $k$ optical switches. Each server is connected to all the $k$ optical switches and in each of them it is connected to both an input and an output port. Each switch is configured to implement a matching between the input and the output ports of the servers, see Figure~\ref{figure_physical_setup}.
Since each server is connected to both input and output sides, the optical switches effectively induce a degree $k$ bipartite graph with $2n$ nodes (two nodes per server). Each optical switch corresponds to a matching in our cache.
It is dynamic as we can insert and evict connections from the switch, but we try to minimize these reconfigurations since they are costly (involve shifting mirrors, and downtime).

At this point we clarify that there are two ``kinds'' of optical switching architectures. The one which we model, as explained, is based on off-the-shelf commodity switches and is sometimes referred to as Optical Circuit Switching (OCS). Each switch is a separate box, and each box, at any time, implements a matching between its ports. We use $k$ switches and connect every server to every switch, so this architecture induces $k$ matchings at any time.
%
To add a connection between two servers we have to choose through which box we want to do it (choose a matching to insert it to) and then reconfigure the matching implemented by this particular box to include this edge. The other kind of switching is known as Free Space Optics (FSO) where every transmitter can point towards any receiver. When each server is 
connected to $k$ transmitters and $k$ receivers we get the standard connection caching setting. This is \emph{not} the architecture that we model here. See Table~1 of \cite{projecToR2016} for several references and their architecture types.

\begin{figure}[t]
	\centering
	\includegraphics[width=0.94\textwidth]{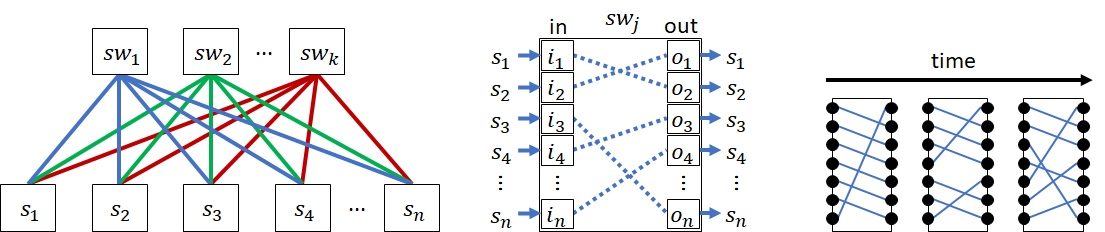}
	\caption{\small{The physical topology that motivates our problem: $n$ servers $s_1,\ldots,s_n$, each is connected to the in/out ports of $k$ optical switches ${sw}_1$ through ${sw}_k$. Each switch ${sw}_j$ uses mirrors to switch optical links, effectively inducing an in/out permutation, which may change over time at a reconfiguration cost of $1$ for each new pairing. Abstractly, we get a bipartite graph with $n$ nodes on each side (one per server), and each permutation is a matching that caches links.}}
	\label{figure_physical_setup}
\end{figure}

Several cost models considering both communication and adjustment cost were suggested for this setting \cite{BMatchingPhysicalDetails}. We choose to work with arguably the simplest model of paying $1$ for an insertion of a new edge (formally defined in Section~\ref{section_model}). This simple model already captures  the qualitative properties of the problem. We note that the competitive results shown here can be adapted (up to constant factors) to a more complicated cost model that has additional communication costs per request. We believe that our combinatorial abstraction of this setting is natural and will find additional applications.

Here is a detailed summary of our results.

\smallskip
\noindent
\textbf{Our contributions:}

\begin{enumerate}
    \item We define a new caching problem, ``Caching in Matchings'' (Problem~\ref{problem_matching_vanilla}), on a bipartite graph with $n$ nodes on each side.\footnote{The problem makes sense on a general graph as well.}
    In this problem, the cache is a union of $k$ matchings. When we insert an edge we pick the matching to insert it to and evict edges from this matching if necessary.

    \looseness=-1 
    \item We show that the competitive ratio of Caching in Matchings depends not only on the cache size  $k$ as is common for caching problems, but also on the number of nodes in the network $n$. One might argue that since we define the cache to be $k$ matchings, its size is $\Theta(nk)$ rather than $k$, so the dependency on $n$ is not surprising. But such an argument also applies to  Connection Caching~\cite{ConnectionCaching} and in that problem the competitive ratio does not depend on $n$. In other words, Caching in Matchings is provably harder than Connection Caching.\footnote{In terms of the architecture, we show that the FSO architecture has a better competitive ratio than the OCS architecture.} Specifically we prove the following.
    \begin{enumerate}
        \item An $\Omega(\max(\frac{n}{k},k))$ lower bound on the competitive ratio of deterministic algorithms, and we give a deterministic algorithm with $nk$ competitive ratio. For $k = O(1)$ this gives a tight bound of $\Theta(n)$ on the competitive ratio.
        
        \item In contrast, in the randomized case we have a larger gap. We describe an $O(n \log k)$ competitive algorithm and prove a lower bound of $\Omega(\max ( 1, \log \frac{n}{k^2 \log k} ) \cdot \log k)$ on the competitive ratio. This bound is $\Omega(\log n)$ for any $k$, and can get as worse as $\Omega(\log^2 n)$, for example if $k = n^{1/3}$. This is in contrast to other caching problems whose randomized competitive ratio is logarithmic.\footnote{Throughout the paper, where it matters, our logarithms are to base $2$.}
    \end{enumerate}

    \item We show that resource augmentation of almost-twice as many matchings, specifically $2k-1$ for the algorithm versus $k$ for the optimum, allows us to eliminate the dependence on $n$. Specifically, we show a deterministic $O(k)$ competitive algorithm and a randomized $O(\log k)$ competitive algorithm for this case. Furthermore, with  $2(1+\alpha)k$ matchings we get a deterministic $O(1+\frac{1}{\alpha})$ competitive algorithm. We also show that a single extra matching already helps by allowing us to ``trade'' $\sqrt{n}$ for $\sqrt{k}$ in the competitive ratio. Concretely and more generally, with $h \ge 1$ extra matchings we get a deterministic $O(n^{1/2} (k/h)^{3/2})$ and a randomized $O(n^{1/2} (k/h)^{1/2} \log \frac{2k+h}{h})$ competitive algorithms. Moreover, it is even possible to reduce the dependence on $n$ to polylogarithmic at the cost of higher polynomial dependency on $k$, which is beneficial for small $k$. Concretely, following~\cite{MultistepVizing2023}, we get a deterministic $O \Big (\frac{k^7 \log n}{h} \min(k, \log n) \Big)$ and a randomized $O\Big (\big (\frac{k \log k}{h} \big )^6 \log \frac{2k+h}{h} \log^9 n \Big )$ competitive algorithms.
\end{enumerate}

\smallskip

Our problem is a special case of a more general problem of convex body chasing in $L_1$. Bhattacharya~\etal~\cite{L1BodyChasing2023} gave a  fractional algorithm for this body chasing problem with packing and covering constraints. Their fractional algorithm requires a slight resource augmentation. For a few special cases, they show how to round their fractional solution to an integral solution that does not use additional resources. Our problem is another interesting test-case of this general setting (see Appendix~\ref{subsection_appendix_randomization_thoughts}).

\smallskip
Our full list of results is summarized in Table~\ref{table_all_results_matchings}. The rest of the paper is structured as follows. Section~\ref{section_model} formally defines the model, the notations that we use, and the caching problems. Section~\ref{section_matchings_caching} studies in depth the Caching in Matchings problem (Problem~\ref{problem_matching_vanilla}). Section~\ref{section_related_work} surveys related work on caching and coloring problems, and in Section~\ref{section_conclusions} we conclude and list a few open questions. Section~\ref{section_appendix} serves as an appendix that contains deferred proofs, and a few additional discussions.

\section{Model and Definitions}
\label{section_model}
In the following we formally define two caching problems of interest, the premise of each of them is a graph with a set $V$ of $n$ nodes. Every turn, a new edge is requested. If it is already cached, we have a ``hit'' and no cost is paid. Otherwise, we have a ``miss'', and the edge must be brought into the cache at a cost of $1$, possibly at the expense of evicting other edges. In fact, the problem that arises from~\cite{BMatchingPhysicalDetails} consists of a bipartite graph in which each server $v$ is associated with two nodes $v^{in} \in V^{in}$ and $v^{out} \in V^{out}$, modeling its receiving and sending ports, respectively. Each of $v^{in}$ and $v^{out}$ can be incident to one edge in each matching.\footnote{We note that technically, the physical switch can be configured with links of the form $(v^{in},v^{out})$, but it makes no sense and practically such requests do not exist. However, our algorithms can deal with all possible requests, and our lower bounds are proven without relying on such requests, so we  ignore this nuance onward.} Formally, the problem is as follows.

\begin{problem}[Caching in Matchings]
\label{problem_matching_vanilla}
Requests arrive for edges $(u,v) \in V^{in} \times V^{out}$. The cache $M$ is a union of $k$ matchings. When a requested edge is missing from all the matchings, an algorithm must fetch it into one of the matchings (possibly evicting other edges from this matching). In addition, the algorithm may choose to add any edge to the cache at any time (while maintaining the cache's restrictions), the cost of adding an edge to the cache is $1$. It is not allowed to move an edge between matchings, but an edge may be evicted and immediately re-fetched into a new matching.
\end{problem}

\begin{remark}
\label{remark_recolor_cost}
There are other caching models in which reorganizing the cache is free, such as \cite{RestrictedCaching2014,CompanionCaching2002}. 
In our model reorganizing the matchings incurs a cost.
This is because we model a setting where changing the cache (physical links) is slow. In other cases accessing the slow memory is the costly operation.
%
\end{remark}

We use the terminology of coloring edges when discussing Caching in Matchings (Problem~\ref{problem_matching_vanilla}). Recoloring an edge implies that we evict it, and then immediately fetch it back into a different matching according to the new color of the edge. Recoloring is not free, but has the same cost as standard fetching. This models, for example, the physical setting in which such a rearrangement requires reconfiguring the link in a different optical switch.

\begin{problem}[Connection Caching~\cite{ConnectionCaching}]
\label{problem_connections_caching_vanilla}
Requests arrive for edges $(u,v) \in V^{in} \times V^{out}$. The cache $M$ is a set of edges such that every node is of degree at most $k$ in the sub-graph induced by $M$. When a requested edge is missing from $M$, an algorithm must fetch it (possibly evicting other edges). In addition, the algorithm may choose to add any edge to the cache at any time (while maintaining the degrees at most $k$), the cost of adding an edge to the cache is $1$.
\end{problem}

\begin{remark}
\label{remark_lazy_is_fine}
Note that in both problems that we defined, an algorithm is allowed to add (fetch) and remove (evict) additional edges. Technically, it is not strictly necessary because a non-lazy algorithm can always be simulated by a lazy version that fetches an edge only when it is actually needed. This is also true for the offline optimum. That being said, we will describe non-lazy algorithms for Caching in Matchings, that recolor edges, to simplify the presentation.
\end{remark}

To emphasize the difference between the problems see Figure~\ref{figure_cannot_add_link_example}, which shows the difference on bipartite graphs, as well as on general graphs (for the generalized problem).

\begin{figure}[ht]
    \centering
    \begin{subfigure}[t]{0.27\textwidth}
        \centering
	\includegraphics[width=0.5\textwidth]{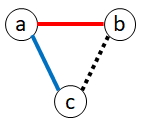}
    \end{subfigure}
    \hspace{5mm} 
	\begin{subfigure}[t]{0.27\textwidth}
        \centering
	\includegraphics[width=0.5\textwidth]{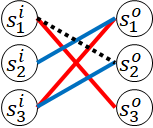}
    \end{subfigure}
    
	\caption{\small{An example of the difference between connection caching and caching in matchings. (left) With $k=2$ max degree and $n=3$ nodes, all the connections can be cached simultaneously, but not in $2$ matchings, red and blue. (right) Bipartite example: Caching the edge $(s^i_1,s^o_2)$ (i/o for in/out) is not possible without changing the matchings (red: $\{ (s^i_1,s^o_3),(s^i_3,s^o_1) \}$; blue: $\{ (s^i_3,s^o_2),(s^i_2,s^o_1) \}$), although both $s^i_1$ and $s^o_2$ only have a single connection.}}
	\label{figure_cannot_add_link_example}
\end{figure}

The objective of an online algorithm is to minimize the number of fetched edges. We are interested in the competitive-ratio of our algorithms.

\begin{definition}[Cost, Competitive Ratio]
\label{definition_competitive_ratio}
Consider a specific caching problem. Let $A$ be an online algorithm that serves requests, and let $\sigma$ be a sequence of requests.  We denote by $A(\sigma)$ the execution of $A$ on $\sigma$, and $cost(A(\sigma))$ for the cost of $A$ when processing $\sigma$.

We denote by $OPT(\sigma)$ the optimum (offline) algorithm to serve the sequence, or simply $OPT$ when $\sigma$ is clear from the context. If there exist functions of the problem's parameters (in our case: $k$ and $n$) $c = c(n,k)$ and $d = d(n,k)$ such that $\forall \sigma: cost(A(\sigma)) \le c \cdot cost(OPT(\sigma)) + d$ then we say that $A$ is $c$-competitive. Note that $\sigma$ may be arbitrarily long, so the ``asymptotic ratio'' is indeed $c$.

Randomized competitiveness is against an oblivious adversary: For every request sequence $\sigma$ fixed independently of $A$'s internal random coins, $\mathbb E[cost(A(\sigma))] \le c \cdot cost(OPT(\sigma))+d$ where the expectation is over $A$'s internal randomness.
\end{definition}

\begin{remark}
\label{remark_cost_of_opts}
Denote the optima for Caching in Matchings and Connection Caching by $OPT_m$ and $OPT_c$, respectively. Since $OPT_m$ implicitly 
maintains a connections cache as required by Connection Caching  (ignore the colors), then for any sequence of edge requests $\sigma$, $cost(OPT_c(\sigma)) \le cost(OPT_m(\sigma))$.
\end{remark}

\section{Caching in Matchings}
\label{section_matchings_caching}
In this section we study the problem of Caching in Matchings (Problem~\ref{problem_matching_vanilla}). We summarize the results of this section in Table~\ref{table_all_results_matchings}. We start with upper bounds (Section~\ref{sub_section_matchings_upper_bounds_bipartite}), then lower bounds (Section~\ref{sub_section_matchings_lower_bounds}). Then we study resource augmentation (Section~\ref{sub_section_matchings_upper_bounds_resource_augmentation}). Some additional discussion on randomization is deferred to the appendix (Section~\ref{subsection_appendix_randomization_thoughts}).

\begin{table}[!ht]
    \begin{center}
        \begin{tabular}{|c|c|c|c|c|c|}

        \hline
        Result & & Deterministic & Randomized & Notes \\
        \hline
        \hline


        Thm.~\ref{theorem_upper_bound_vanilla} & &  $\le nk$ & $O(n \log k)$ &  The standard scenario \\
        \hline
        \hline
        
        Thm.~\ref{theorem_lower_bound_randomized_lg_n_lg_k} & ${*}$  & . & $\Omega(\max ( 1, \log \frac{n}{k^2 \log k} ) \cdot \log k)$ & .  \\
        \hline
        
        Cor.~\ref{corollary_lower_bounds_pure_n} & ${*}$ & $\Omega(\max(\frac{n}{k},k))$ & $\Omega(\log n)$ & Due to Theorems~\ref{theorem_lower_bound_deterministic}+\ref{theorem_lower_bound_randomized_lg_n_lg_k} \\
        \hline

        Cor.~\ref{corollary_lower_bounds_pure_n} & ${*}$ & . & $\Omega(\epsilon \cdot \log n \cdot \log k)$ & $k = O(n^{1/2 - \epsilon})$; Due to Theorem~\ref{theorem_lower_bound_randomized_lg_n_lg_k} \\
        \hline
        \hline

        Cor.~\ref{corollary_augmented_algorithms}(\ref{item_basic}) & & $O(n^{1/2} (k/h)^{3/2})$ & $O(n^{1/2} (k/h)^{1/2} \log \frac{2k+h}{h})$ & RA: $k+h$ for $1 \le h \le k$ \\

        \hline
        Cor.~\ref{corollary_augmented_algorithms}(\ref{item_2023paper}) & ${*}$ & $O \Big (\frac{k^7 \log n}{h} \min(k, \log n) \Big)$ & $O\Big (\big (\frac{k \log k}{h} \big )^6 \log \frac{2k+h}{h} \log^9 n \Big )$ & RA: $k+h$ for $1 \le h \le k$ \\

        \hline
        Cor.~\ref{corollary_augmented_algorithms}(\ref{item_noN}) & ${*}$ & $\le k$ & $O(\log k)$ & RA: $2k-1$ \\

        \hline
        Cor.~\ref{corollary_augmented_algorithms}(\ref{item_2024paper}) & ${*}$ & . & $O ( \alpha^4 \log k)$ & \begin{tabular}{@{}c@{}}RA: $(1+O(\frac{1}{\alpha}))k$ \\ for $k \ge \Theta(\alpha^4 \log n)^{\Theta(\alpha \log \alpha)}$\end{tabular} \\
        
        \hline
        Thm.~\ref{theorem_resource_augmentation_det_even_more} & ${*}$ & $O(1 + \frac{1}{\alpha})$ & . & RA: $(2+\alpha)k$ for $\alpha > 0$ \\

        \hline

        \end{tabular}
    \end{center}
    \caption{\small{Our bounds on the competitive ratio for Caching in Matchings (Problem~\ref{problem_matching_vanilla}), for $2 \le k < n$. If $k=1$ or $k \ge n$ optimality is trivial. Results marked with ${*}$ also apply to general graphs. ``RA: $x$'' is an abbreviation for Resource-Augmentation, i.e., the algorithm has more matchings ($x$) than the optimum ($k$). Randomization is against an oblivious adversary.}}
    \label{table_all_results_matchings}
\end{table}

\subsection{Upper Bounds for Bipartite Graphs}
\label{sub_section_matchings_upper_bounds_bipartite}
In this section we prove upper bounds on the competitive ratio of algorithms for Caching in Matchings, focusing on the non-trivial case of $2 \le k \le n-1$. Indeed, if $k=1$ there are no eviction-decisions to make so the only (lazy) algorithm is the optimal one. The other extreme case of $k \ge n$ in bipartite graphs is also easy since we can just cache the entire graph: Number the nodes $0$ to $n-1$ on each side, and use matching $i$ to store edges from node $j$ to $i+j$ modulo $n$.

Our general technique is to reduce the problem of Caching in Matchings to Connection Caching. Our algorithm, $A_m$, will run a Connection Caching algorithm $A_c$ with cache parameter $k$ to insert requested edges into the cache. Then, layered on top of $A_c$, we have the ``coloring component'' of $A_m$ that chooses the color of the new edge, and also recolors existing edges in order to produce a proper Caching in Matchings algorithm. $A_m$ can be thought of as an edge coloring algorithm in the dynamic graph setting, and in this context $A_c$ is the adversary that tells $A_m$ which edges are inserted and which are removed (with a guarantee of bounded degree $k$). As a consequence we would like to use algorithms that are efficient in terms of recoloring, to achieve the best competitive results. Unfortunately, since edge coloring of graphs of bounded degree $k$ may require $k+1$ colors by Vizing's theorem~\cite{Vizing1964}, the dynamic graph coloring literature studies this coloring problem while typically allowing more than $k$ colors. The number of extra colors ranges from $k+1$ colors \cite{MisraGries1992_EdgeColoring,Bernshteyn2022SmallRecoloring}, to $(1+\epsilon)k$ colors \cite{DynamicEdgeColoring2019,MultistepVizing2023,NibblingCycles2024}, to $2k-1$ colors \cite{BarenboimSublinearColoring,bhattacharya2018dynamic}, and sometimes even more \cite{DynamicGraphColoringBarba2019, ImprovedDynamicColoring} (the last citations actually study vertex coloring). Extra colors correspond to resource augmentation, which we study later in Section~\ref{sub_section_matchings_upper_bounds_resource_augmentation}.

\begin{remark}
\label{remark_competitiveness_of_connections}
There are known algorithms that are $k$ competitive deterministic and $O(\log k)$ competitive randomized for Connection Caching, as studied in~\cite{ConnectionCaching}.
\end{remark}

\begin{remark}
\label{remark_multiply_competitive_factors}
Due to Remark~\ref{remark_cost_of_opts} and Remark~\ref{remark_competitiveness_of_connections}, it suffices to analyze the cost ratio between $A_m$ and $A_c$. A ratio of $\rho$ implies a $\rho \cdot k$ deterministic and a $O(\rho \cdot \log k)$ randomized competitive algorithms for Caching in Matchings.
\end{remark}

\begin{theorem}
\label{theorem_upper_bound_vanilla}
There exist $nk$ deterministic and $O(n \log k)$ randomized competitive algorithms in bipartite graphs for Caching in Matchings.
\end{theorem}

\begin{proof}
By Remark~\ref{remark_competitiveness_of_connections} and Remark~\ref{remark_multiply_competitive_factors}, it suffices to show that $A_m$ pays no more than $n$ times compared to $A_c$. Whenever an edge $(u,v)$ is requested, $A_m$  has it cached if and only if $A_c$ has it cached. Therefore when $A_m$ has a miss, so does $A_c$. To accommodate for the edge, $A_c$ ensures that $u$ and $v$ are both of degree $k-1$ before $(u,v)$ is inserted. Now consider how many edge recolorings are required from $A_m$. Nodes $u$ and $v$ each have at least one free color. If both have some common free color $c$, we are done. Otherwise, $u$ has $c_1$ free and $v$ has $c_2 \ne c_1$ free. Let $P_u$ and $P_v$ be the $(c_1,c_2)$ bi-colored paths that originate in $u$ and $v$ respectively. $P_u$ and $P_v$ must be disjoint because the graph is bipartite and does not contain odd cycles. Flipping the colors ($c_1 \leftrightarrow c_2$) for each edge on either $P_u$ or $P_v$ enables $A_m$ to insert and color $(u,v)$. Since $P_u$, $P_v$ and $(u,v)$ form a simple path in a graph with $2n$ nodes, by flipping the shorter bi-colored path, $A_m$ colors at most $n$ edges when inserting $(u,v)$.
\end{proof}

\subsection{Lower Bounds}
\label{sub_section_matchings_lower_bounds}
Caching in Matchings is a generalization of caching, if we restrict the requests to  edges of a single fixed node. Observe, therefore, that any $c$-competitive online algorithm for Caching in Matchings with $2 \le k < n$ satisfies $c = \Omega(\log k)$. Moreover, if the algorithm is deterministic then $c \ge k$. The following lower bounds depend on $n$ as well as $k$. These bounds hold for the non-trivial case of $2 \le k < n$, in bipartite graphs, and therefore also hold for general graphs. Theorem~\ref{theorem_lower_bound_deterministic} is proven later in this section, the proof of Theorem~\ref{theorem_lower_bound_randomized_lg_n_lg_k} is deferred to Appendix~\ref{appendix_sub_section_matchings_lower_bounds_proofs}.

\begin{restatable}{theorem}{theoremLowerBoundDeterministic}
\label{theorem_lower_bound_deterministic} Any
\emph{deterministic} Caching in Matchings algorithm is $\Omega(\frac{n}{k})$ competitive.
\end{restatable}

\begin{restatable}{theorem}{theoremLowerBoundRandLgNTimesLgK}
\label{theorem_lower_bound_randomized_lg_n_lg_k}
\emph{Any} Caching in Matchings algorithm is $\Omega(\log \frac{n}{k^2 \log k} \cdot \log k)$ competitive. (This expression is meaningful for $\log \frac{n}{k^2 \log k} \ge 1$, otherwise, we have $\Omega(\log k)$ due to caching.)
\end{restatable}

\begin{corollary}
\label{corollary_lower_bounds_pure_n}
\emph{Any} online algorithm for Caching in Matchings with $2 \le k < n$ is $\Omega(\log n)$ competitive. Moreover, if $k = O(n^{1/2 - \epsilon})$ for some constant $\epsilon > 0$, we get that any online algorithm for Caching in Matchings is  $\Omega(\epsilon \cdot \log n \cdot \log k)$ competitive. If the algorithm is deterministic then the competitive ratio is $\Omega(\max\{\frac{n}{k},k\})$.
\end{corollary}

\begin{proof}
The deterministic claim follows from the initial observation and Theorem~\ref{theorem_lower_bound_deterministic}. In the general case (randomized), we get $\Omega(\log n)$ from the maximum between the observation (when $k \ge n^{1/3}$) and Theorem~\ref{theorem_lower_bound_randomized_lg_n_lg_k} (when $k < n^{1/3}$). The 
$\Omega(\epsilon \cdot \log n \cdot \log k)$  bound 
 follows from Theorem~\ref{theorem_lower_bound_randomized_lg_n_lg_k}: If $k \le c \cdot n^{1/2-\epsilon}$ for some constant $c$ then $\log \frac{n}{k^2 \log k} > \log \frac{n^{2\epsilon}}{c^2 \log (cn)} = 2\epsilon \log n - 2\log c - \log \log (cn) = \Omega(\epsilon \cdot \log n)$.
\end{proof}

We prove the lower bounds in a setting that is closer to dynamic graph coloring. Specifically, we define Problem~\ref{problem_add_remove_edges} below, where we control which edges must be cached both by the algorithm and the optimum. We prove (Lemma~\ref{lemma_new_model_is_fine_for_lower_bound} below, proven in Appendix~\ref{appendix_sub_section_matchings_lower_bounds_proofs}) that lower bounds for algorithms for Problem~\ref{problem_add_remove_edges} imply lower bounds for Caching in Matchings, and then study lower bounds for Problem~\ref{problem_add_remove_edges}.

\begin{problem}
\label{problem_add_remove_edges}
Given a graph with $n$ vertices, we get a sequence of actions that define a subset of edges at any time. Each action either adds a missing edge or deletes an existing edge. We are guaranteed that at any point in time the graph induced by existing edges, denote it (or the set of edges) by $G$, has a proper $k$-edge-coloring. An algorithm, online or $OPT$, must maintain $k$ matchings, denote their union by $M$, such that $G$ is a subgraph of $M$. For every edge that is added to a matching, the algorithm pays $1$.
\end{problem}

\looseness=-1
Note that Problem~\ref{problem_add_remove_edges} is similar but not equivalent to dynamic edge coloring. On one hand a dynamic edge coloring algorithm that recolors $O(C)$ edges per update is not necessarily $O(C)$  competitive for Problem~\ref{problem_add_remove_edges}. The reason for this is that we allow $M$ to contain $G$. By maintaining an edge in $M$ we can avoid paying for it when it is inserted again. For example, in the proof of  Theorem~\ref{theorem_lower_bound_deterministic} an algorithm may do $O(1)$ worst-case recolorings per step, but its competitive ratio is $\Omega(\frac{n}{k})$ since $OPT$ stores in $M$ extra edges that this online algorithm keeps paying for. On the other hand, an algorithm that is $O(C)$ competitive for Problem~\ref{problem_add_remove_edges} does not give a dynamic edge coloring algorithm that recolors $O(C)$ edges per update, even amortized, because it could be that both $OPT$ and the algorithm pay a lot per edge update on some sequence, and while the ratio is $O(C)$, the absolute cost is large.

\begin{restatable}{lemma}{lemmaReductionToProblemStaticEdges}
\label{lemma_new_model_is_fine_for_lower_bound}
A lower bound of $C$ on the competitive ratio of an online algorithm for Problem~\ref{problem_add_remove_edges} implies a lower bound of $C$ on the competitive ratio of an online algorithm for Caching in Matchings.
\end{restatable}

We now focus on deriving lower bounds for  Problem~\ref{problem_add_remove_edges}.
We define a \emph{road} gadget (Definition~\ref{definition_road}) which is a connected component with a large diameter, that is also very restricted in the way it can be colored. A road is constructed from \emph{brick} sub-gadgets (Definition~\ref{definition_brick_short}), each of size $\Theta(k)$ nodes and $\Theta(k^2)$ edges. By connecting $r \ge 2$ roads together we get the \emph{$r$-road} gadget, whose structure enforces those roads to be colored in a distinct and different way.

\begin{definition}[Brick]
\label{definition_brick_short}
A \emph{colorless-brick} is a union of $k$ perfect matchings in a bipartite graph with the following structure. Each side has $w$ nodes where $w$ is the unique power of $2$ that satisfies $\frac{w}{2} < k \le w$. Number the nodes on each side  $0,\ldots,w-1$, and number the colors  $0,\ldots,k-1$. The matching of color $c$ matches node $i$ with node $i \oplus c$ where $\oplus$ is the bitwise exclusive-or. See Figure~\ref{figure_brick_and_road-brick} for an example. Note that every color $c$ indeed defines a matching that is in fact a permutation of order $2$ when $c \ne 0$ and the identity when $c=0$, and that $v \oplus c \ne v \oplus c'$ for any two colors $c \ne c'$ so the matchings are all disjoint. 
When we remove an edge from a colorless-brick, we get a \emph{brick} whose color is associated with the color of the non-perfect matching. The two nodes of degree $k-1$ are the \emph{endpoints} of the brick.
\end{definition}

\begin{figure}[ht]
	\centering
	\begin{subfigure}[t]{.14\textwidth}
        \centering
        \includegraphics[width=\textwidth]{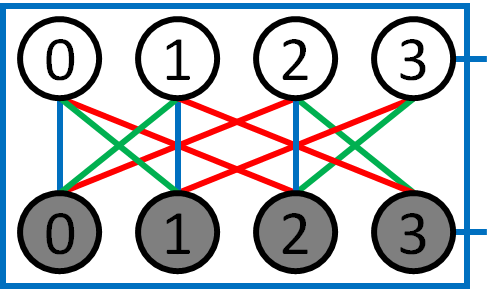}
        \caption{Brick}
        \label{figure_brick_and_road-brick} 
    \end{subfigure}
    \hspace{5mm} 
    \begin{subfigure}[t]{.21\textwidth}
        \centering
        \includegraphics[width=\textwidth]{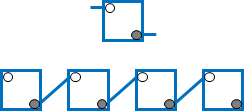}
        \caption{Road}
        \label{figure_brick_and_road-road}
    \end{subfigure}
    \hspace{5mm} 
    \begin{subfigure}[t]{.22\textwidth}
        \centering
        \includegraphics[width=\textwidth]{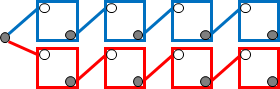}
        \caption{$r$-Road for $r=2$}
        \label{figure_brick_and_road-r_road}
    \end{subfigure}

	\caption{\small{Visualization of a \emph{brick} (Definitions~\ref{definition_brick_short}), a \emph{road} and an \emph{$r$-road} (Definition~\ref{definition_road}). (a) A brick for $k=3$ ($w=4$). The number of each node is written, and the colors are as follows: blue (0), green (1), red (2). Thus, for example, $1 \oplus red = 3$. We removed a single edge, blue $(3,3)$, thus the brick is blue. (b) A schematic way to draw a brick (top) and a road of length $4$ (bottom) which is a  chain of bricks connected to each other by their endpoints. The color of a road is well-defined by the color of its bricks. (c) An $r$-road for $r=2$, of length $4$. The node that connects to both roads is its hub. Nodes are colored by gray and white according to their side in the bipartite graph.}}
	\label{figure_brick_and_road}
\end{figure}

\begin{definition}[Road, $r$-road]
\label{definition_road}
A \emph{road} of length $d \ge 1$ is an edge colored graph obtained by connecting a sequence of $d$ bricks. Each brick is connected by an edge to the next brick in the sequence. The edge connecting two consecutive bricks is adjacent to an endpoint of each brick. Note that the color of two connected bricks must be the same since they
must agree on their free color which is the color of the edge which connects them. Therefore, we define the \emph{color of a road} to be the color of its bricks. A road has two \emph{ends}, which are the endpoints of its first and last bricks.
We also refer to $r$ ($2 \le r \le k$) roads of the same length $d$ that are all connected to a single shared node as an \emph{$r$-road} of length $d$. The shared node is its \emph{hub}. The edges of the hub all have different colors, therefore all the roads of an $r$-road have different colors. See Figure~\ref{figure_brick_and_road} for examples.
\end{definition}

\begin{lemma}
\label{lemma_brick_special}
Given a brick $B$ of  color  $c_1$, and a new color $c_2 \ne c_1$, it is always possible to recolor $3$ edges to change the color of $B$ to $c_2$.
\end{lemma}

Note that when we recolor $B$, it no longer satisfies the $\oplus$-property of Definition~\ref{definition_brick_short}, but for convenience we still consider it as a brick.  
This would not affect our arguments below (by more than a constant factor) since we will make sure to always return to the original coloring (undo) before recoloring again.

\begin{proof}
Denote by $u$ one endpoint of the brick. By definition of the matching scheme, the other endpoint is $u \oplus c_1$, and when we restrict the graph to edges of colors $c_1$ and $c_2$, we find that the path between $u$ and $u \oplus c_1$ is of length $3$: $u \oplus c_1 \to u \oplus c_1 \oplus c_2 \to u \oplus c_2 \to u$. Therefore, it suffices to flip the color of these three edges from $c_2$ to $c_1$ and vice versa.
\end{proof}

In the remainder of this section we prove the deterministic lower bound, and a simpler but weaker version of the randomized lower bound. The more involved randomized lower bound is proven in Appendix~\ref{appendix_sub_section_matchings_lower_bounds_proofs}, using the same gadgets.

\begin{proof}[Proof of Theorem~\ref{theorem_lower_bound_deterministic}]
We prove the lower bound for Problem~\ref{problem_add_remove_edges}. Then the theorem follows by Lemma~\ref{lemma_new_model_is_fine_for_lower_bound}. We present an adversarial construction against a given algorithm $ALG$.

We begin by setting aside one special node $u$ to serve as a $2$-road hub, and divide the rest of the vertices into bricks. We construct from these bricks the longest possible road, of length $N = \Theta(\frac{n}{k})$. We number the bricks in order, from $1$ to $N$, and denote the edge between bricks $i$ and $i+1$ by $(i,i+1)$. Let $L = \floor{\frac{N}{3}}$. Initially we insert all the edges of all the bricks, without the edges connecting the bricks. These edges are never deleted. Our sequence has as many steps as we like as follows, based on the state of $ALG$, see also Figure~\ref{figure_lower_bound_det_visualized}:

\begin{enumerate}
    \item Simple step: If there exist consecutive bricks $i$ and $i+1$ of different colors, we insert the edge $(i,i+1)$. This forces $ALG$ to recolor at least one of the bricks and pay $\Omega(1)$ (it could pay more if it recolors more bricks or does other actions). We then delete $(i,i+1)$.

    \item Split step: Otherwise, all the bricks of $ALG$ have the same color. We insert all the edges between bricks $1$ through to $L$, and between $N+1-L$ through to $N$. We also insert  edges from $u$ to bricks $1$ and $N+1-L$. These insertions construct a $2$-road with $u$ as its hub and guarantee different colors for bricks $1$ to $L$ compared to bricks $N+1-L$ to $N$. $ALG$ must recolor at least $L = \Omega(N)$ bricks. We then delete the edges that we inserted.
\end{enumerate}

\begin{figure}[ht]
	\centering
    \includegraphics[width=0.93\textwidth]{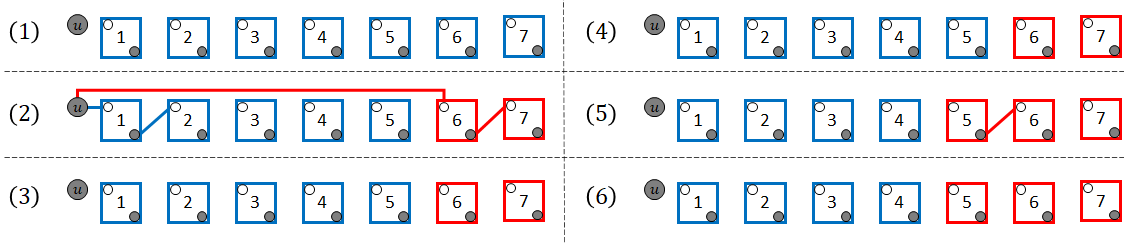}
	
	\caption{\small{Visualization for the proof of Theorem~\ref{theorem_lower_bound_deterministic}, for $N=7$, $L=\floor{\frac{N}{3}}=2$. (1)-(3) A \emph{split step} temporarily creates a $2$-road of length $L$ out of the first and last bricks, with $u$ as the hub, to guarantee consecutive bricks with different colors. (4)-(6) A \emph{simple step} finds two consecutive bricks with different color (here: $5$ and $6$), and inserts temporarily the edge between them to enforce recoloring one of them.}}
	\label{figure_lower_bound_det_visualized}
\end{figure}

In simple terms, we maintain a hole in the road which is where the color of the bricks changes (there could be multiple holes). In every step we request this hole, and once the hole disappears, the split action re-introduces a hole back near the middle of the road.\footnote{This idea is similar to the way one can prove a deterministic lower bound for $k$-server~\cite{Borodin_book}, by always requesting a server-less location. The analogy stops here, since in our case sometimes there is no hole.}

Now let us analyze the costs. If the sequence contains $m$ simple steps and $s$ split steps, then $cost(ALG) = \Omega(m + s \cdot N)$. For $OPT$, we define a family of strategies $B_i$ for $L < i < N+1-L$, to bound its cost. We define $B_i$ to store all the edges of all the bricks, all the edges connecting them, and the edge that connects $u$ to the first brick, except for the edges that connect brick $i$ to its neighbours, paying an initial $O(N \cdot k^2)$ cost.  $B_i$ colors all the bricks from $1$ to $i$ in one color, and all the bricks from $i+1$ to $N$ in another color. Whenever a simple step happens in $(i-1,i)$ or in $(i,i+1)$, $B_i$ simply recolors brick $i$ to have the same color as the neighbour to which it connects, and also inserts the connecting edge. Simple steps at other locations do not affect $B_i$. When a split step happens, $B_i$ pays exactly $2$: It inserts the edge that connects $u$ to brick $N+1-L$ instead of the edge $(N-L,N+1-L)$. When the split step ends, it undoes the change, re-inserting $(N-L,N+1-L)$ instead of the edge of $u$. Since $L < i < N+1-L$, $B_i$ does not have to recolor any other edge.

Thus, if we denote by $m_i$ the number of simple steps that insert the edge $(i,i+1)$, then $cost(OPT) \le cost(B_i) = O(m_i + m_{i-1} + s + N \cdot k^2)$. Now:
$\frac{N}{3} \cdot cost(OPT) \le (N-2L) \cdot cost(OPT) \le \sum_{i=L+1}^{N-L}{cost(B_i)} = O(m + N \cdot s + N^2 k^2)$.
Note that $N^2 k^2 = O(n^2)$. Thus by extending the sequence such that $m+s = \Omega(n^2)$, we get: $\frac{N}{3} \cdot cost(OPT) = O(m + N \cdot s) = O(cost(ALG))$. Therefore, $\frac{cost(ALG)}{cost(OPT)} = \Omega(N)$. Recall that $N = \Theta(\frac{n}{k})$, and the claim follows.
\end{proof}

When proving the deterministic lower bound (Theorem~\ref{theorem_lower_bound_deterministic}) we heavily relied on determinism to know where to find neighbouring bricks of different colors. In the randomized case, we may not know where colors mismatch. Instead, we use a different and weaker construction for the randomized case. Relying on Yao's principle~\cite{YaoPrinciple} (see also \cite{Borodin_book}), we define a distribution over sequences that is hard for any algorithm.

The following Lemma~\ref{lemma_lower_bound_randomized_lg_n_over_k} is a weaker but also simpler version of Theorem~\ref{theorem_lower_bound_randomized_lg_n_lg_k} (proven in Appendix~\ref{appendix_sub_section_matchings_lower_bounds_proofs}). Proving this lemma demonstrates our main technique. 

\begin{lemma}
\label{lemma_lower_bound_randomized_lg_n_over_k}
\emph{Any} Caching in Matchings algorithm with $2 \le k < n$  matchings is $\Omega(\log \frac{n}{k})$ competitive.
\end{lemma}

\begin{proof}
We prove the lower bound for Problem~\ref{problem_add_remove_edges}. Then the theorem follows by Lemma~\ref{lemma_new_model_is_fine_for_lower_bound}. For convenience, since any action of fetching or recoloring can be lazily postponed to the next request for adding an edge to $G$, we assume that the algorithm does nothing else when an edge is removed.

Given a fixed $k$, we divide the nodes into $r$-roads of length $d_0$. We aim to have $2^h$ $r$-roads in total, for $h$ as large as possible. We can bound $h$ from below by noting that a brick requires at most $4k$ nodes, a road of length $d_0$ requires at most $4k \cdot d_0$ nodes ($d_0$ bricks), and an $r$-road requires at most $4k \cdot d_0 \cdot r + 1$ nodes ($r$ roads and a hub). So we get that $h \ge \floor{\log \frac{n}{4k \cdot d_0 \cdot r + 1}}$. Simplified, we get $h = \Omega(\log \frac{n}{k \cdot d_0 \cdot r})$. Eventually we will choose $r=2$ and $d_0=1$ to get $h = \Omega(\log \frac{n}{k})$.

We construct the distribution over request sequences  in phases. A phase begins with $2^h$ $r$-roads. Then, during each phase we have $h$ rounds, numbered from $i=1$ to $i=h$. In each round we pair the $r$-roads, and merge each pair into a single twice-longer $r$-road. Note that in round $i$ there are $2^{h-i}$ pairs of $r$-roads whose length is $d_i = d_0 \cdot 2^{i-1}$. Once we get to a final single $r$-road of length $d_h$, we delete the edges that were used to connect the roads of length $d_0$, and insert back hub edges to re-form the initial $r$-roads of length $d_0$. Then a new phase begins.

We explain later how exactly a pair of $r$-roads of length $d_i$ is merged, for now just assume that $OPT$ pays $O(r)$ and that any algorithm pays $\Omega(r \cdot d_i)$ in expectation for such a merge, and let us analyze the competitive ratio. When a phase begins, $OPT$ can choose a consistent color for each road of the $r$-roads such that no further recoloring is necessary during this phase, at a cost of $O(2^h \cdot r \cdot d_0)$ by recoloring $O(1)$ edges in each brick (according to Lemma~\ref{lemma_brick_special}) and each edge that connects to a brick to the desired color. Then, throughout the rounds it pays additional $\sum_{i=1}^{h}{2^{h-i} \cdot O(r)} = O(2^h \cdot r)$. Finally, when a phase ends it pays $O(2^h \cdot r)$ more to re-attach hubs when re-creating the initial $r$-roads of length $d_0$, and $O(2^h \cdot r \cdot d_0)$ more to undo any recoloring made when the phase began.\footnote{It is necessary to revert to the exact initial coloring before the next phase because Lemma~\ref{lemma_brick_special} for recoloring bricks requires a very specific coloring scheme.} Overall, $OPT$ pays per phase $O(2^h \cdot r \cdot d_0)$. In contrast, $ALG$ pays at least $\sum_{i=1}^{h}{2^{h-i} \cdot \Omega(r \cdot 2^{i-1} \cdot d_0)} = \Omega(h \cdot 2^h \cdot r \cdot d_0)$.

Let $t$ be the number of phases. The one-time initialization cost is $c_0 = \Theta(k^2 \cdot (2^h \cdot r \cdot d_0))$ for inserting $\Theta(k^2)$ edges per brick. Therefore, the competitive ratio that we get is $\frac{\mathbb{E}[cost(ALG)]}{cost(OPT)} \ge
\frac{c_0 + t \cdot \Omega(h \cdot 2^h \cdot r \cdot d_0)}{c_0 + t \cdot O(2^h \cdot r \cdot d_0)}$. For $t = \Omega(k^2)$ we can neglect $c_0$ and get that $\frac{\mathbb{E}[cost(ALG)]}{cost(OPT)} = \Omega(h)$. Recall that $h = \Omega(\log \frac{n}{k \cdot r \cdot d_0})$, so we choose $r=2$ and $d_0=1$ to maximize the competitive ratio and get $\Omega(\log \frac{n}{k})$, as claimed.

It remains to explain how we merge a pair of $r$-roads of length $d$, $X$ and $Y$, see Figure~\ref{figure_multi_road_simple_merge} for a visual example for $r=2$ and $k=4$. We have $r$ iterations, where iteration $i$ cuts the $i$th road of $Y$ away from the hub, and extends a uniformly random not yet extended road of $X$. $OPT$ pays at most $r$ for the newly introduced $r$ edges because it can refrain from recoloring roads. As for $ALG$, observe that it must recolor a road if the colors of the extended road and its extension do not match. For the first two roads that we combine (one from each $r$-road), there is a probability of at most $\frac{1}{r}$ for the colors to agree (the probability is maximized if $X$ and $Y$ use the same colors for their roads, out of the $k$ possible colors). More generally, in the $i$th iteration there is a probability of at most $\frac{1}{r+1-i}$ for the colors to agree, maximized if the remaining roads of $Y$ share their colors with the not yet extended roads of $X$. Thus $ALG$ recolors in expectation at least $\sum_{i=1}^{r}{(1-\frac{1}{r+1-i})} = r - H_r$ roads throughout the process, where $H_r$ is the $r$th harmonic number. For $r \ge 2$, this amounts to a cost of $\Omega(r \cdot d)$ recolorings.
\end{proof}

\begin{figure}[ht]
	\centering
    \begin{subfigure}[t]{.22\textwidth}
        \centering
        \includegraphics[width=\textwidth]{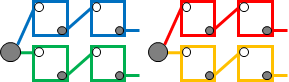}
        \caption{Initial state}
        \label{figure_multi_road_simple_merge1}
    \end{subfigure}
    \hspace{5mm} 
	\begin{subfigure}[t]{.22\textwidth}
        \centering
        \includegraphics[width=\textwidth]{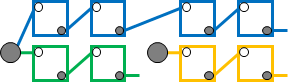}
        \caption{One road extended}
        \label{figure_multi_road_simple_merge2}
    \end{subfigure}
    \hspace{5mm} 
    \begin{subfigure}[t]{.22\textwidth}
        \centering
        \includegraphics[width=\textwidth]{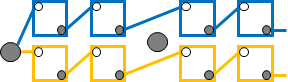}
        \caption{Merge complete}
        \label{figure_multi_road_simple_merge3}
    \end{subfigure}
    
		\caption{\small{Visualization of merging a pair of $2$-roads of length $d=2$ to a single twice longer $2$-road. In this example $k\ge 4$. Initially the $2$-roads are disjoint. Then, we cut a road from the right $2$-road and extend another road in the left $2$-road. If necessary, the algorithm recolors one or more of the roads. Then we do the same for the remaining road. In the end, the hub of the cut-down $2$-road is a node of degree $0$.}}
	\label{figure_multi_road_simple_merge}
\end{figure}

\subsection{Upper Bounds with Resource Augmentation}
\label{sub_section_matchings_upper_bounds_resource_augmentation}

In this section we study upper bounds with resource augmentation. That is, we assume that the optimum still has $k$ matchings, but our algorithm has more. Interestingly, it dramatically improves the competitive ratios, in both deterministic and randomized settings.

Recall our general approach and notations: our caching in matchings algorithm $A_m$ implements a component of an edge-coloring algorithm, over a component of a connection caching algorithm which we denote by $A_c$. By Remark~\ref{remark_multiply_competitive_factors}, we can divide our attention between connection caching and dynamic edge-coloring. Concretely, given $h \ge 1$ extra matchings, we maintain connection caching with $k' \equiv k + h_1$ connections per node, and maintain edge-coloring of a graph of maximum degree $k'$, with $k' + h_2$ colors, such that $h_1 + h_2 = h$. We choose $h_2 \ge 1$ because a single extra color yields a dramatic improvement. We use either $h_2 = \ceil{\frac{h}{2}}$ or $h_2 = h$, Corollary~\ref{corollary_augmented_algorithms} summarizes our choices.

We begin by listing three important facts about caching and connection caching algorithms, which we  use as $A_c$. Concretely, Lemma~\ref{lemma_resource_augmentation_connections_det} details a deterministic algorithm for Connection Caching, and Lemma~\ref{lemma_resource_augmentation_rand} together with Theorem~\ref{theorem_decoupling_competitive} yield a randomized algorithm for Connection Caching in Corollary~\ref{corollary_resource_augmentation_connections_rand}.

\begin{lemma}[Corollary 8 of \cite{ConnectionCaching}]
\label{lemma_resource_augmentation_connections_det}
There is a deterministic Connection Caching algorithm with cache of size $r$, that is $\frac{2r}{r-k+1}$-competitive against the optimum with cache of size $k \le r$.
\end{lemma}

\begin{lemma}[Section 2.2 of \cite{NealYoungRandAugmented}]
\label{lemma_resource_augmentation_rand}
Let $r$ be the cache size of the randomized caching algorithm \emph{MARK}~\cite{MARK_FIAT1991}, and let $k$ be the cache size of the optimum. Then \emph{MARK} is: $O(\log r)$-competitive if $r=k$; $O(\log \frac{r}{r-k})$-competitive if $\frac{e-1}{e} r<k<r$; and, $2$-competitive if $k \le \frac{e-1}{e} r$.
\end{lemma}

\begin{restatable}[Theorem 7 of \cite{ConnectionCaching}]{theorem}{theoremDecouplingCCtoCaching}
\label{theorem_decoupling_competitive}
Let $A$ be a $c(r,k)$-competitive caching algorithm, with additive term $\delta$, where $r$ and $k$ are the cache sizes of the algorithm and the optimum, respectively. Then there is a $2 \cdot c(r,k)$-competitive algorithm for Connection Caching, with additive term $|V| \cdot \delta$ where $|V|$ is the number of nodes in the graph.
\end{restatable}

For completeness, we give the explicit reduction and proof of Theorem~\ref{theorem_decoupling_competitive} in Appendix~\ref{appendix_subsection_misc_proofs}.

\begin{corollary}
\label{corollary_resource_augmentation_connections_rand}
There exists a randomized connection caching algorithm, with cache size $r$ compared to $k$ of the optimum, that is $O(\log r)$-competitive if $r=k$; $O(\log \frac{r}{r-k})$-competitive if $\frac{e-1}{e} r<k<r$; and, $4$-competitive if $k \le \frac{e-1}{e} r$.
\end{corollary}

Next, we list several results for dynamic edge coloring.

\begin{lemma}[Greedy, Folklore]
\label{lemma_trivial_greedy}
Let $G$ be a dynamic graph with the guarantee that its maximum degree is at most $k$ at any time. Then we can maintain for it a $2k-1$ edge-coloring without needing to recolor any edge.
\end{lemma}
\begin{proof}
When $(u,v)$ is inserted, both $u$ and $v$ are of degree at most $k-1$, thus each has at least $k$ free colors, and they must have at least one common free color that we can use.
\end{proof}

\begin{restatable}{lemma}{lemmaSqrtNK}
\label{lemma_sqrtnk}
Let $G$ be a dynamic \emph{bipartite} graph with the guarantee that its maximum degree is at most $k$ at any time, and let $h \ge 1$. We can maintain a deterministic $(k+h)$-edge-coloring of $G$ in amortized $O(\sqrt{nk/h})$ recolorings per insertion.
\end{restatable}

The proof of Lemma~\ref{lemma_sqrtnk} is deferred to Appendix~\ref{appendix_sub_section_matchings_augmentation_proofs}. It is like the proof of Theorem~\ref{theorem_upper_bound_vanilla} but we only have few edges from each extra color, such that we can recolor bi-chromatic paths quickly. 
Periodically, we recolor the graph using only $k$ colors and amortize this work over several operations.

The following results
are particularly useful when $k$ is small. The high probability in Theorem~\ref{theorem_multi_vizing_dynamic_graphs} below is for bounding the running time, not for getting a proper coloring.

\begin{theorem}
\label{theorem_multi_vizing_existence}
Let $G' \equiv G \cup \{ e \}$ be a graph with maximum degree $k$, such that $G$ is $(k+1)$-edge-colored, and $e$ is uncolored. Then there is a $(k+1)$-edge-coloring of $G'$ which recolors only $N$ edges in $G$ where $N = O(k^7 \log n)$  by \cite{MultistepVizing2023} (Theorem 3), or $N = (k+1)^6 \log^2 n$ by \cite{Bernshteyn2022SmallRecoloring} (Corollary 6.4).
\end{theorem}

\begin{theorem}[Theorem 6 of \cite{MultistepVizing2023}]
\label{theorem_multi_vizing_dynamic_graphs}
Let $G$ be a dynamic graph such that its maximum degree never exceeds $k$. Then there exists a fully-dynamic algorithm that maintains a $\ceil{(1+\epsilon)k}$-edge-coloring with $O(\epsilon^{-6} \log^6 k \log^9 n)$ worst-case update time with high probability.
\end{theorem}

The paper~\cite{NibblingCycles2024} gives an efficient randomized edge-coloring for sufficiently large $k$.

\begin{theorem} [\cite{NibblingCycles2024}]
\label{theorem_nibble_improved}
Let $G$ be a dynamic graph such that its maximum degree never exceeds $k$. If $k \ge (100 \alpha^4 \log n)^{30\alpha \log \alpha}$, there is a fully-dynamic algorithm maintaining a $(1+O(\frac{1}{\alpha}))k$-edge-coloring with $O(\alpha^4)$ edge recolorings in expectation per update.
\end{theorem}

Finally, we combine the various results of connection caching and edge coloring to derive the following competitive algorithms.

\begin{corollary}
\label{corollary_augmented_algorithms}
Given resource augmentation of extra $1 \le h = O(k)$ matchings, that is $k+h$ for the algorithm versus $k$ for OPT, the following competitive algorithms for caching in matchings exist:
\begin{enumerate}
    \item \label{item_basic} $O(n^{1/2} (k/h)^{3/2})$ deterministic and $O(n^{1/2} (k/h)^{1/2} \log \frac{2k+h}{h})$ randomized competitive algorithms in bipartite graphs.
    
    \item \label{item_2023paper} $O \Big (\frac{k^7 \log n}{h} \min(k, \log n) \Big)$ deterministic and $O\Big (\big (\frac{k \log k}{h} \big )^6 \log \frac{2k+h}{h} \log^9 n \Big )$ randomized competitive algorithms in general graphs.
    
    \item \label{item_noN} If $h=k-1$ we can remove the dependence on $n$, yielding $k$ deterministic and $O(\log k)$ randomized competitive algorithms in general graphs.
    
    \item \label{item_2024paper} $O(\alpha^4 \log k)$ randomized competitive algorithm, where $\alpha = O(\frac{k}{h})$ and provided that $k \ge (100 \alpha^4 \log n)^{30\alpha \log \alpha}$.
\end{enumerate}
\end{corollary}

\begin{proof}
We use the augmentation to have extra $h_2$ colors for the coloring component, where $h_2 = h$ for Part-(\ref{item_noN}), and $h_2 = \ceil{\frac{h}{2}}$ for the rest. Part-(\ref{item_basic}) is by Lemma~\ref{lemma_sqrtnk} with Lemma~\ref{lemma_resource_augmentation_connections_det} and Corollary~\ref{corollary_resource_augmentation_connections_rand}. Part-(\ref{item_2023paper}) is by Theorem~\ref{theorem_multi_vizing_existence} with Lemma~\ref{lemma_resource_augmentation_connections_det}, and by Theorem~\ref{theorem_multi_vizing_dynamic_graphs} with Corollary~\ref{corollary_resource_augmentation_connections_rand}. Regarding the randomized part, Theorem~\ref{theorem_multi_vizing_dynamic_graphs} guarantees recoloring that is cheap with high probability. We can choose the constants such that the failure probability is $\le \frac{1}{n^2}$, and fully recolor the graph if the cheap method fails. The expected number of recolorings is negligibly affected, and proper edge coloring is guaranteed. Notice that we set $\epsilon = \frac{h}{k}$.
Part-(\ref{item_noN}) is by Lemma~\ref{lemma_trivial_greedy}.
Part-(\ref{item_2024paper}) is by Theorem~\ref{theorem_nibble_improved} and Corollary~\ref{corollary_resource_augmentation_connections_rand}. Theorem~\ref{theorem_nibble_improved} only guarantees a high probability to fit within the palette size, but the cited algorithm can first be run virtually on each update, to determine failure by observing the resulting palette size. Upon failure we recolor everything from scratch. Even though this adds $O(nk)$ recolorings, it rarely happens so this added cost does not affect the competitive ratio.
\end{proof}

\begin{remark}
\label{remark_on_augmentation_corollary}
Choosing $h_2 = \ceil{\frac{h}{2}}$ in Corollary~\ref{corollary_augmented_algorithms} divides the augmentation into equal halves and is good enough if we do not optimize the constants, since essentially both caching and coloring components ``benefit'' from $\Theta(h)$ augmentation.

In Part-(\ref{item_2024paper}) it is simplest to think of $\alpha$ as a constant, in which case the requirement $k \ge f(n,\alpha)$ for the function $f$ given in the statement, requires $k = \Omega(poly (\log n))$. However, $\alpha$ can also depend on $k$, as long as there are values of $k$ that satisfy $k \ge f(n,\alpha(k))$. Observe that because $(\log n)^{\log n / \log \log n} = n$, for $k \le n$ it must be that $\alpha = O(\frac{\log n}{\log \log n})$ or else the inequality cannot be satisfied. Then in particular $\alpha = o(\log n)$, and $k^{1 / {\tilde{\Theta}(\alpha)}} \ge \log n$ (for the appropriate constants) implies $k \ge f(n, \alpha(k))$.
Crudely simplified for the sake of a clean example, if $k^\frac{1}{\alpha} \ge \log n$ we could choose $\alpha = \frac{\log k}{2 \log \log k}$, and still have a non-empty range of applicable $k$ values.
\end{remark}

Finally, we improve the competitive ratio further with a larger resource augmentation.

\begin{restatable}{theorem}{theoremResourceAugmentationLargeDet}
\label{theorem_resource_augmentation_det_even_more}
Given  $k+h$ matchings to  the algorithm compared to only $k$ matchings to the optimum, for $h \ge k-1$, there is a deterministic algorithm that is $2(1 + \frac{k-1}{\floor{\frac{h+3-k}{2}}})$-competitive for Caching in Matchings. In particular, with $h = (1 + \alpha) k$ extra matchings we get a competitive ratio of $O(1 + \frac{1}{\alpha})$.
\end{restatable}

The proof is in Appendix~\ref{appendix_sub_section_matchings_augmentation_proofs}. It follows from Lemma~\ref{lemma_trivial_greedy} and Lemma~\ref{lemma_resource_augmentation_connections_det}.

\begin{corollary}
\label{corollary_explicit_augmentations}
Consider the Caching in Matchings problem where the optimum is given $k$ matchings. There is a $6$-competitive algorithm that uses $3k-3$ matchings ($h = 2k-3$).
\end{corollary}

\textbf{A note on the running times:} 
This work focuses on competitive analysis and therefore we do not attempt to optimize polynomial running time. We note that  the algorithms in Corollary~\ref{corollary_augmented_algorithms}(\ref{item_basic}),(\ref{item_noN}) take $O(n)$ time: maintaining connection caching takes $O(k)$ time, and the coloring algorithms in Lemma~\ref{lemma_trivial_greedy} and Lemma~\ref{lemma_sqrtnk} naively take time of $O(\Delta)$ to find a color and $O(\rho)$ to recolor a path of length $\rho$. The randomized algorithms of Corollary~\ref{corollary_augmented_algorithms}(\ref{item_2023paper}),(\ref{item_2024paper}) also run in polynomial time per update.

\section{Related Work}
\label{section_related_work}

\textbf{Caching Problems:}
Caching problems have been studied in many variants and cost models. \emph{Connection caching} is the caching variant closest to our problem. We presented it in a centralized setting, Cohen~\etal~\cite{ConnectionCaching} introduced it in a distributed setting. Albers~\cite{ConnectionCachingGeneralized} studies generalized connection caching. Bienkowski~\etal~\cite{BMatchingPaper} study connection caching in a cost model that is similar to that of caching with rejections~\cite{RejectionCaching}. Another related variant is \emph{restricted caching}~\cite{RestrictedCaching2014,CompanionCaching2002} where not every page can be put into every cache slot. Buchbinder~\etal~\cite{RestrictedCaching2014} study the case where each page $p$ has a subset of cache slots in which it can be cached. In our problem we also have a restriction of similar flavour, implied by the separation into matchings. We note that the cost model in \cite{RestrictedCaching2014,CompanionCaching2002} only counts cache-misses, while we also pay for rearranging the cache.


\smallskip
\noindent
\textbf{Coloring Problems:} As mentioned in Section~\ref{section_matchings_caching}, an efficient dynamic edge coloring that uses a small number of colors can be useful for competitive analysis. Subsection~\ref{sub_section_matchings_upper_bounds_resource_augmentation} covers results which we use for our advantage. The coloring literature studies the tradeoff between the number of colors, amount of recoloring (sometimes called \emph{recourse}), and the running time of the algorithms. Some algorithms require a bound $\Delta$ on the maximum degree of the dynamic graph, while others are adaptive with respect to the maximum degree in their running time or recoloring. Literature on vertex coloring also exists, but reducing edge coloring to vertex coloring by coloring the line-graph is too wasteful in the number of colors, whether this number is parameterized by $\Delta$, or by the arboricity of the graph as in \cite{DynamicColoringImplicitExplicit2020}.
Works on maintaining dynamically an \emph{implicit coloring}~\cite{christiansen2022fullydynamicArboricity,DynamicColoringImplicitExplicit2020} cannot apply to our case because the matchings form an explicit coloring. Azar~\etal~\cite{AzarCompetitiveVertexRecoloring} study dynamic vertex coloring in the context of competitive analysis.

\looseness=-1
El-Hayek~\etal~\cite{BMatching2023} are motivated by the same architecture as us. They solve a problem of dynamically maximizing the size of a $k$-edge-colorable subgraph of a dynamic graph.

\smallskip
\noindent
\textbf{Linear Programming and Convex Body Chasing in $L_1$:} 
The aforementioned caching problems, like many other combinatorial problems, can be formulated as a linear program~\cite{Buchbinder_book}. This line of research led to the development of competitive algorithms for weighted and  generalized caching~\cite{Generalized_Caching_SODA2012,Generalized_Caching,GeneralizedCaching_Buchbinder,WeightedCaching_Buchbinder}. A recent result of Bhattacharya~\etal~\cite{L1BodyChasing2023} uses linear programming with packing and covering constraints to formulate and frame the problem as convex body-chasing in $L_1$. They give a fractional algorithm that requires a slight resource augmentation, along with some rounding schemes to get randomized algorithms for specific problems. Our problem can be thought of as another special case of the problem considered by \cite{L1BodyChasing2023}, see Appendix~\ref{subsection_appendix_randomization_thoughts} for this formulation and further details.

\section{Conclusions and Future Work}
\label{section_conclusions}
In this paper we studied the online Caching in Matchings problem, in which we receive requests for edges in a graph and need to maintain a cache of the edges which is a union of $k$ matchings. The problem abstracts some hardware architecture in which a datacenter is enhanced with reconfigurable optical links. Interestingly, we proved that the Caching in Matchings problem is inherently harder than the similar-looking Connection Caching problem and other caching problems.
Specifically, its competitive ratio depends not only on the number of matchings $k$ (``cache size'') but also on the number of nodes in the graph. Our randomized lower bound rules out an $O(\log n)$ competitive algorithm, and the best competitive ratio we can hope for is $O(poly(\log n))$. Our lower bound for deterministic algorithms is linear in $n$ for a fixed $k$.

We derived our algorithms by running a coloring algorithm that maintains a coloring of the cache of a connection caching algorithm. This approach is simple to describe and analyze, but inherently multiplies the competitive ratios of the two algorithms. It is natural to ask whether a ``direct'' algorithm for Caching in Matchings exists, and if so, whether it improves the competitive ratio?

Regarding resource augmentation of $h \ge 1$ extra matchings, we see that there are two interesting ``discontinuities''. First, immediately for $h=1$ the competitive ratio drops to $poly(k, \log n)$, in particular ``breaking'' the deterministic lower bound. Second, there seems to be a point in which the competitive ratio becomes independent of $n$. It clearly happens for $h=k-1$, and even sooner if $k$ is large enough (revisit Corollary~\ref{corollary_augmented_algorithms}). 
These ''discontinuities'' beg the following two questions. First, is there an $1 \le h < k-1$ and an algorithm that uses $h$ extra matchings with a competitive ratio of $O(poly(\log n))$ for any $2 \le k < n$? Differently phrased, can we achieve a competitive ratio that is $poly(\log k, \log n)$ instead of $poly(k, \log n)$? Second, is it possible to remove the dependence on $n$ using fewer than $h=k-1$ extra matchings for any $k$, and if so how small can $h$ be?

A natural generalization would be to study upper bounds for Caching in Matchings in general graphs. When $k=1$ optimality is still trivial, 
and when $k = n$, by Vizing's theorem,  we are also optimal since we can edge-color the full $n$-clique with $n$ colors. In fact, for $n \ge 2$, $k=n-1$ colors are sufficient if and only if $n$ is even (Lemma~\ref{lemma_even_graph_factorization} in Appendix~\ref{appendix_subsection_misc_proofs}).
In the non-trivial regime $2 \le k < n$, there exists a naive deterministic $O(n^2 k^2)$ competitive algorithm (Lemma~\ref{lemma_k_upper_bound_general} in Appendix~\ref{appendix_subsection_misc_proofs}).
Resource augmentation of an extra matching ($k+1$ colors) dramatically reduces the competitive ratio to $O(nk)$ (deterministic) and $O(n \log k)$ (randomized) by allowing us to update the coloring of the graph when a new edge is inserted according to a single step of the Misra-Gries algorithm~\cite{MisraGries1992_EdgeColoring},\footnote{The Misra-Gries algorithm edge-colors an uncolored graph with $m$ edges and $n$ nodes in $m$ iterations. In each iteration it colors an edge and fixes the colors of previously colored edges, by recoloring $O(n)$ of them in $O(n)$ time. In our case, the graph is always fully colored up to the newly requested edge, so we get the competitive ratio of a connection caching algorithm, multiplied by $O(n)$.} or to $O(poly(k,\log n))$ as in Corollary~\ref{corollary_augmented_algorithms}(\ref{item_2023paper}). It is an interesting question whether the problem without augmentation is indeed that much harder in general graphs. General graphs also provide additional difficulties, such as the fact that finding a minimum edge coloring for $k \ge 3$ is generally NP-complete~\cite{ChromaticIndexNPC}.

\pagebreak 

\section{Appendix: Deferred Proofs and Discussions}
\label{section_appendix}

\subsection{Caching in Matchings Lower Bounds (Proofs)}
\label{appendix_sub_section_matchings_lower_bounds_proofs}

\lemmaReductionToProblemStaticEdges*

\begin{proof}
Let $\algProbMatchings$ be a $c$-competitive algorithm for Caching in Matchings (Problem~\ref{problem_matching_vanilla}), with some additive term $d$. We may assume that $\algProbMatchings$ is lazy by Remark~\ref{remark_lazy_is_fine}, that is, upon successful request it changes nothing in its cache. We show how to derive from it an algorithm $\algProbFixedEdges$ that is $c$-competitive for Problem~\ref{problem_add_remove_edges}, which proves the claim. We will also use corresponding subscripts $\optProbMatchings$ and $\optProbFixedEdges$ for the optimum of each problem (with respect to a given sequence). 

Given a sequence $\tau$, Algorithm $\algProbFixedEdges$ takes its decisions while processing $\tau$ by simulating $\algProbMatchings$ on a sequence $\sigma$ which is constructed as follows. We traverse $\tau$ in order, and whenever an edge is inserted, we add to $\sigma$ a batch of requests which is a concatenation of $r=nk$ identical subsequences, each subsequence contains all the edges currently in $G$ (in some arbitrary order). When an edge is deleted, we do nothing.

 We now specify $\algProbFixedEdges$ such that $cost(\algProbFixedEdges(\tau))) \le cost(\algProbMatchings(\sigma))$ by having $\algProbFixedEdges$ maintain its state such that it ``jumps'' between ``check-points'' in the state of $\algProbMatchings$.

$\algProbFixedEdges$ maintains a flag that says whether its state is synchronized with that of $\algProbMatchings$ (initially true) and works as follows. When an edge is inserted by $\tau$, $\algProbFixedEdges$ feeds the corresponding batch of requests in $\sigma$ to $\algProbMatchings$, and by the end of this batch one of two things happens: either (1) the state of $\algProbMatchings$ provides a proper coloring of $G$, or (2) $G$ is not fully in its state. In case (1), if the flag was on (synchronized state), $\algProbFixedEdges$ changes its state by replaying the changes made by $\algProbMatchings$. By definition of this case, it ends up with a proper coloring of $G$ while paying the same cost. Still in case (1), if the flag was off, then $\algProbFixedEdges$ copies the state of $\algProbMatchings$ by the end of the batch, at a cost of at most $|G| \le |M| \le \frac{nk}{2} = \frac{r}{2}$. This also turns on the flag, and the cost can be charged to the most recent going off-sync when the flag turned off. Indeed, in case (2), we know that during the whole batch $\algProbMatchings$ did not have a proper coloring of $G$, which means that in each of the $r$ rounds it paid at least $1$ for a missing edge, for a total of at least $r$. Then $\algProbFixedEdges$ can just pick a valid state at a cost of at most $|G| \le |M| \le \frac{nk}{2} = \frac{r}{2}$, falling off-sync, and ``banking'' the remaining $\frac{r}{2}$ to be paid to get back into sync next time case (1) occurs. Overall, we argued that $cost(\algProbFixedEdges(\tau)) \le cost(\algProbMatchings(\sigma))$. This is true in the deterministic case, and also for the randomized case for every fixing of the random coins.

Now observe that $cost(\optProbMatchings(\sigma)) \le cost(\optProbFixedEdges(\tau))$. Indeed, $\optProbMatchings$ may simulate the behavior of $\optProbFixedEdges$ by making changes to its own state at the beginning of each batch in $\sigma$. In conclusion, we get that $cost(\algProbFixedEdges(\tau)) \le cost(\algProbMatchings(\sigma)) \le c \cdot cost(\optProbMatchings(\sigma)) + d \le c \cdot cost(\optProbFixedEdges(\tau)) + d$ in the deterministic case, or similarly $\mathbb{E}[cost(\algProbFixedEdges(\tau))] \le \mathbb{E}[cost(\algProbMatchings(\sigma))] \le c \cdot cost(\optProbMatchings(\sigma)) + d \le c \cdot cost(\optProbFixedEdges(\tau)) + d$ in the randomized case (recall that $c$ and $d$ were defined at the beginning of the proof).
\end{proof}

\theoremLowerBoundRandLgNTimesLgK*
\begin{proof}
This proof uses a similar high-level construction as the one used to prove Lemma~\ref{lemma_lower_bound_randomized_lg_n_over_k}, and the only difference is in the way we merge pairs of $r$-roads. Here we set $r=k$. Assume for now that when merging two $k$-roads of length $d$, $OPT$ pays $O(k \log k)$ and any algorithm pays $\Omega(k \log k \cdot d)$ in expectation.

With this in mind, revisit the competitive analysis: when a phase begins $OPT$ pays $O(2^h \cdot k \cdot d_0)$ to recolor bricks to their desired color, it then pays $\sum_{i=1}^{h}{2^{h-i} \cdot O(k \log k)} = O(2^h \cdot k \log k)$ in the merging rounds, and finally it pays $O(2^h \cdot k \cdot d_0)$ to restore the original $k$-roads for the next phase. Overall, its cost is $O(2^h \cdot k \cdot (d_0 + \log k))$. In comparison, $ALG$ pays for recoloring, in expectation, at least $\sum_{i=1}^{h}{2^{h-i} \cdot \Omega(k \log k \cdot 2^{i-1} \cdot d_0)} = \Omega(h \cdot 2^h \cdot k \log k \cdot d_0)$. Assuming a sequence with $t = \Omega(k^2)$ phases, we get that the competitive ratio is $\frac{\mathbb{E}[cost(ALG)]}{cost(OPT)} = \Omega(\frac{h \cdot 2^h \cdot k \log k \cdot d_0}{2^h \cdot k \cdot (d_0 + \log k)}) = \Omega(\frac{h \cdot \log k \cdot d_0}{d_0 + \log k})$. To maximize the expression we balance and choose $d_0 = \ceil{\log k}$, getting $\Omega(h \log k)$. We determine $h$ as before, except that the complicated merging technique requires a reusable extra node, so we have that $h \ge \floor{\log \frac{n-1}{4k^2 \cdot d_0 + 1}}$. Simplified, and with $d_0 = \ceil{\log k}$, we get $h = \Omega(\log \frac{n}{k^2 \log k})$, therefore the competitive ratio is $\Omega(\log \frac{n}{k^2 \log k} \cdot \log k)$.

\begin{figure}[ht]
	\centering
    \begin{subfigure}[t]{.22\textwidth}
        \centering
        \includegraphics[width=\textwidth]{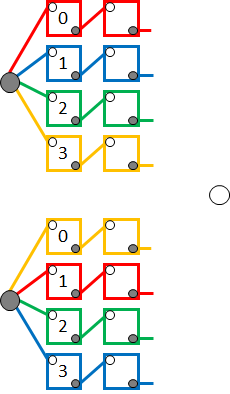}
        \caption{Initial state}
        \label{figure_rroad_complex_merge_step1}
    \end{subfigure}
    \hspace{2mm} 
	\begin{subfigure}[t]{.22\textwidth}
        \centering
        \includegraphics[width=\textwidth]{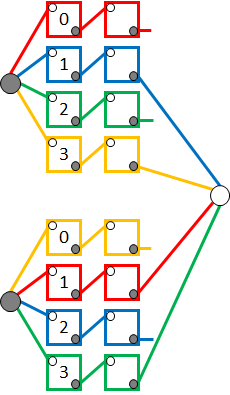}
        \caption{Step 1}
        \label{figure_rroad_complex_merge_step2}
    \end{subfigure}
    \hspace{2mm} 
    \begin{subfigure}[t]{.22\textwidth}
        \centering
        \includegraphics[width=\textwidth]{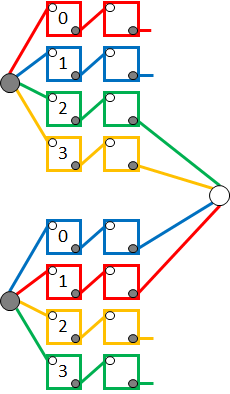}
        \caption{Step 2}
        \label{figure_multi_road_complex_merge3}
    \end{subfigure}
    \hspace{2mm} 
    \begin{subfigure}[t]{.22\textwidth}
        \centering
        \includegraphics[width=\textwidth]{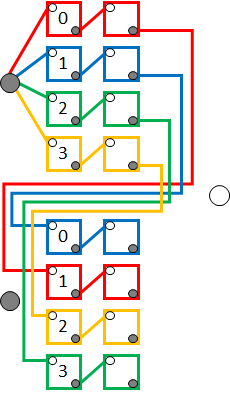}
        \caption{Final merge}
        \label{figure_multi_road_complex_merge4}
    \end{subfigure}
    
		\caption{\small{Visualization of merging a pair of $k$-roads, for $k=4$, of length $d=2$ to a single twice longer $k$-road, by ``negative information''. (a) The roads are numbered from $0$ to $3$, with their number written inside their first brick. There are $\log k = 2$ steps. (b) In the first step we temporarily connect roads $\{1,3\}$ (least significant bit $1$) of the top $k$-road with either roads $\{0,2\}$ or $\{1,3\}$ of the bottom $k$-road, to a shared hub (the white node). In this example we connect $\{1,3\}$, and as a result, roads $2$ and $3$ of the bottom $k$-road were recolored. (c) In the second step we temporarily connect roads $\{2,3\}$ (second bit is $1$) of the top $k$-road with either roads $\{0,1\}$ or $\{2,3\}$ of the bottom $k$-road to a shared hub. In this example we connect $\{0,1\}$, and as a result roads $0$ and $2$ of the bottom $k$-road were recolored. (d) Finally there is a round of ``positive information'' in which we simply extend each road on the top color-consistently with a road on the bottom. The consistency depends on the choices of the previous steps, and in this example it matches road $x$ on the top with road $x \oplus 1$ on the bottom. The recoloring in this example is such that in the final extension no road is recolored.}}
	\label{figure_k_roads_neg_merge_example}
\end{figure}

Now we explain and analyze in detail the merging of two $k$-roads, denote them by $X$ and $Y$. See Figure~\ref{figure_k_roads_neg_merge_example} for a visual example with $k=4$. For simplicity, let us start with $k$ being an integer power of $2$, say $k=2^\ell$. We start with $\ell$ steps of ``negative information'' in which we reveal roads that are of \emph{different} colors, and do so by connecting the free end of these roads to a new shared hub, denote it by $v$. Concretely, number the roads of $X$ from $0$ to $k-1 = 2^\ell - 1$ and denote by $X_{i,b}$ the roads of $X$ whose $i$th bit is $b$. We define similarly the subsets of roads for $Y$. In round $i$, we connect to $v$ the roads $X_{i,1}$ and $Y_{i,b_i}$ where $b_i$ is chosen uniformly at random between $0$ and $1$. Note that $|X_{i,1}| = |Y_{i,b_i}| = \frac{k}{2}$ so the degree of $v$ is $k$ (legal). When the round ends, we delete the edges of $v$. Finally, in step $\ell+1$ we produce the longer $k$-road with ``positive information'' by cutting the roads of $Y$ from their hub and extending the roads of $X$, according to the unique way which does not contradict the previous $\ell$ steps. This way exists: road $y \in Y$ extends the road $x$ whose binary representation is $x = y \oplus B$ where the bits of $B$ are $b_i$ and $\oplus$ is the bitwise exclusive-or operation.

Let us analyze the costs of $OPT$ and $ALG$. Since $OPT$ knows the correct colors in advance, it can pay at most $1$ per edge that is inserted. We insert $k$ edges per step (even if most of them are later deleted), in a total of $\ell + 1$ steps. This amounts to $O(k \log k)$. We argue that $ALG$ recolors in expectation $\Omega(k)$ roads in each of the first $\ell$ steps. To simplify the analysis, assume that $ALG$ recolors after learning $b_i$ rather than being introduced online each edge of $v$ one by one (which can only hurt $ALG$). Then indeed every road in $X_{i,1}$ has probability $\frac{1}{2}$ to be in a color conflict with a road of $Y_{i,b_i}$, so by the linearity of expectation, we get at least $\frac{1}{2} \cdot \frac{k}{2} = \Omega(k)$ road recolorings (if $ALG$ is ``reasonable'', it recolors $O(k)$ roads per step, so allowing it to be semi-offline did not lose more than a constant factor). Observe that our bound for round $i$ is not affected by previous rounds. So we conclude that $ALG$ pays $\Omega(\ell \cdot k \cdot d) = \Omega(k \log k \cdot d)$ for recoloring in expectation.

The case of $k$ not being a power of $2$ is similar. Each road is still assigned a number, and we regard its binary representation with $\ceil{\lg k}$ bits, but only make $\ell = \floor{\lg k}$ rounds. Note that now $X_{i,1}$ is not necessarily of size $\frac{k}{2}$, but rather might be smaller. The bias is always in favor of $0$ because of how counting works, and it is such that $|X_{i,1}| \ge \frac{k - 2^i}{2}$ ($i=0$ is the least significant bit). So we can still choose $Y_{i,b_i}$ with $b_i$ uniformly random, there is no problem to connect all the roads of $X_{i,1}$ and $Y_{i,b_i}$ to their shared hub. Also, each road in $X_{i,1}$ still has a color conflict with probability $\frac{1}{2}$. The only thing that changes is that the expectation of road recolorings is a little less than before, but we get at least $\frac{1}{2}$ times $\sum_{i=0}^{\ell - 1}{|X_{i,1}|} \ge \ell \cdot \frac{k}{2} - \frac{1}{2} \sum_{i=0}^{\ell - 1}{2^i} > \ell \cdot \frac{k}{2} - \frac{2^\ell}{2} > \frac{(\ell - 1)k}{2}$ road recolorings in expectation for $ALG$, which is still $\Omega(k \log k)$ in total. The analysis of $OPT$ is unchanged, and its total cost is $O(k \log k)$ per merge of a pair of $k$-roads (of any length).

To complete the argument for non-powers of $2$, we should also address the mapping $x = y \oplus B$ which is no longer well-defined because $x$ could fall out of range. For example, $k=5$ (then $\ell = 2$), $B=3$ and $y=4$ yield $x = 7 > k$. Denote $k = 2^\ell + q$. For the roads numbered $0 \le i \le 2^\ell - 1$ we still use the mapping $x = y \oplus B$. It remains to define a bijection for the high-indexed roads which is consistent with the first steps. When we ignore their top $\ell$-bit which is $1$, their index falls in the range $D \equiv [0,q-1]$ (by subtracting $2^\ell$). We will define a bijection $f: D \to D \oplus B$ such that $x \subseteq f(x)$ when $x$ and $f(x)$ are viewed as sets of their $1$-bits. This relation is important for consistency with the negative-information rounds. Initially $f(x) = x$. Now process the $1$ bits of $B$ in any order, and for bit $0 \le j \le \ell - 1$ do this: For every pair $\{x , x \oplus 2^j \}$ such that $x \le x \oplus 2^j$: if $x \in f(D)$ but $x \oplus 2^j \notin f(D)$ we replace $x$ by $x \oplus 2^j$ in $f(D)$. Note that because $D = [0,\ldots,q-1]$ and unpaired values have a $0$ bit, and because operations over bit $j$ don't affect bit $j'$ of the mapping, it cannot be that $x \oplus 2^j \in f(D)$ while $x \notin f(D)$. By construction we get that indeed $x \subseteq f(x)$ and $f(D) = D \oplus B$.\footnote{$x \subseteq f(x)$: Initially $x=f(x)$, and then we only change $0$ bits to $1$ bits. $f(D) = D \oplus B$: In each bit, $B$ affects $D$ like $\oplus$. This is true for the domain as a whole but not for individual values because we don't switch a pair $\{ x, x \oplus 2^j \} \subseteq f(D)$ when $B_j = 1$.} Finally, we pair the remaining $X$-road $2^\ell + x$ with the $Y$-road $2^\ell + (f(x) \oplus B )$. Since $f$ is a bijection from $D$ to $D\oplus B$, these targets are exactly the remaining $Y$ roads $2^\ell + D$, each used once. Moreover, $x \subseteq f(x) = (f(x) \oplus B ) \oplus B$ so every pair is compatible with all the negative-information rounds. To clarify, we matched the roads $x' = 2^\ell + x$ and $y' = 2^\ell + (f(x) \oplus B)$, whose roles in all the rounds are identical to the roads $x \in X$ and $y = (f(x) \oplus B) \in Y$, respectively. Since $x \subseteq y \oplus B$ (viewed as $1$-bits sets), it means that when $x$ was chosen in a round of negative information, $y$ wasn't, so they are compatible (this generalizes the case when $x = y \oplus B$).
\end{proof}

\begin{remark}
\label{remark_notes_on_lower_bound_proof}
A few notes on the proofs of Lemma~\ref{lemma_lower_bound_randomized_lg_n_over_k} and Theorem~\ref{theorem_lower_bound_randomized_lg_n_lg_k}:
\begin{enumerate}
    \item The random choices of the adversary can be boiled-down to the random order of extending roads (in Lemma~\ref{lemma_lower_bound_randomized_lg_n_over_k}) and the bits $b_i$ (in Theorem~\ref{theorem_lower_bound_randomized_lg_n_lg_k}). The $2$-roads and $k$-roads themselves are chosen once, and even the pairings of each merging round may be fixed.

    \item For clarity, we presented it as if we need $2^h$ different hubs, one per $r$-road. In fact, we only need $h+1$ hubs if we reuse them: $h$ of them to maintain an $r$-road for each unique length, and another one for the length in which we currently merge a pair of $r$-roads. This saving is negligible compared to the number of nodes used to compose the roads.
\end{enumerate}
\end{remark}

\subsection{Caching in Matchings Resource Augmentation Upper Bounds (Proofs)}
\label{appendix_sub_section_matchings_augmentation_proofs}

In this section we restate and prove the claims from Section~\ref{sub_section_matchings_upper_bounds_bipartite} that we did not prove there.

\lemmaSqrtNK*
\begin{proof}
Recall the proof of Theorem~\ref{theorem_upper_bound_vanilla}, we will use the same idea of a color-swap on a bi-colored path. The key difference is how we use the $h \ge 1$ extra colors.

First consider $h=1$ and denote the extra color as yellow. We allow at most $y$ yellow edges in the graph, and if we need more, we recolor the whole graph from scratch without using yellow. Such a recoloring is possible because the graph is bipartite and every node is of degree at most $k$. When coloring a newly inserted edge $(u,v)$, both $u$ and $v$ have at least one free color. We have three cases:
\begin{enumerate}
    \item If $u$ and $v$ share a free color, including yellow: Then use this color.

    \item If $u$ does not have a yellow edge and $v$ does (the other case is symmetric): Let $c$ be a free color of $v$, and apply a color-swap of $c$ and yellow with respect to $v$. This makes yellow a free color of $v$. Note that $u$ is unaffected by the color-swap, because the graph is bipartite
    (affecting $u$ implies that the path of the swap 
   closes an odd cycle  with the edge $(u,v)$). Now color $(u,v)$ in yellow.

    \item If both $u$ and $v$ have a yellow edge: Let $c$ be a free color of $u$. Apply a color-swap of $c$ and yellow with respect to $u$. This makes yellow a free color of $u$. If yellow becomes free at $v$ then $c$ was free at $(u,v)$ to begin with and we would have picked the first case, otherwise now apply the previous case ($v$ will still have a yellow edge at this point, even if it is on the path affected by $u$).
\end{enumerate}

We apply up to two color-swaps, each of length $O(y)$ because there are at most $y$ yellow edges in the whole graph. Recall that we might have a global recoloring once we reach $y$ yellow edges. We charge these recolorings to the yellow edges. Formally, we define a potential for the cache which equals $\frac{nk}{y} \cdot i$ when there are $i$ yellow edges. Thus when we accumulate $y$ yellow edges, the potential can pay for the global recoloring. Each insertion of an edge causes $O(y)$ recolorings and increases the potential by at most $\frac{nk}{y}$, due to possibly inserting a yellow edge (our color-swaps never increase the number of yellow edges). We conclude that the amortized cost is $O(y + \frac{nk}{y})$ per insertion. Balancing with $y=\sqrt{nk}$ gives $O(\sqrt{nk})$.

We generalize the previous logic for $h \ge 1$ by allowing each extra color to have at most $y$ edges, and when it fills up we proceed to use the next extra color. Only when all $h$ colors have $y$ edges we invoke a full recoloring. The potential in this case is $\frac{nk}{hy}$ per edge, and the amortized cost is therefore $O(y + \frac{nk}{hy})$. Balancing with $y = \sqrt{nk/h}$ gives $O(\sqrt{nk/h})$.
\end{proof}

\theoremResourceAugmentationLargeDet*
\begin{proof}
Let $\sigma$ be a sequence of requests. Denote an algorithm $A$ with cache parameter $x$ as $A^x$, and use subscripts $m$ for Caching in Matchings and $c$ for Connection Caching. We have $cost(A_m^{2r-1}(\sigma)) = cost(A_c^{r}(\sigma))$ for any integer $r \ge 1$ by considering $r-1$ matchings as resource augmentation, such that we require no recoloring (by Lemma~\ref{lemma_trivial_greedy}). Since this reduction halves the cache parameter, and our algorithm initially has cache of size $k+h$, we use $r=\floor{\frac{k+h+1}{2}}$. If $k+h = 2r$, we do not use one of the colors, on purpose, to ensure using exactly $2r-1$ colors. Taking $A_c$ to be the algorithm that satisfies Lemma~\ref{lemma_resource_augmentation_connections_det}, $cost(A^r_c(\sigma)) \le \frac{2r}{r-k+1} \cdot cost(OPT^k_c(\sigma)) + d$ for some fixed term $d$. By Remark~\ref{remark_cost_of_opts}, $cost(OPT^k_c(\sigma)) \le cost(OPT^k_m(\sigma))$. Plugging everything together we get that 
$cost(A^{2r-1}_m(\sigma)) \le \frac{2r}{r-k+1} \cdot cost(OPT^k_m(\sigma)) + d$, hence $A^{2r-1}_m$ is $\frac{2r}{r-k+1} = 2(1 + \frac{k-1}{r-k+1}) = 2(1 + \frac{k-1}{\floor{\frac{h+3-k}{2}}})$ competitive for Caching in Matchings.
\end{proof}

\subsection{Further Discussion on Randomization for Caching in Matchings}
\label{subsection_appendix_randomization_thoughts}

The gap between lower and upper bounds is not too wide in the deterministic case, when $k$ is small (e.g. $k=O(1)$), but it is exponentially wide with respect to $n$ in the randomized case. It is quite common that randomized algorithms achieve better competitive ratios, and it is particularly known in caching and other online problems (e.g., k-servers), which gives a ``reason'' to believe this may be so for Caching in Matchings as well.

In this section we describe a naive randomized algorithm and prove that it fails, and then proceed to discuss in details a more promising direction based on linear programming formulation.

\begin{lemma}
\label{lemma_RC_bad_competitive}
Consider the following randomized algorithm, $ALG$: When a new edge is requested, if no matching can accommodate it, pick its color uniformly at random and evict conflicting neighbours. For $k=2$, $ALG$ is $\Omega(n)$ competitive.
\end{lemma}

\begin{proof}
We construct a sequence on which the expected cost of this algorithm is high. The construction proceeds in alternating phases as follows. We single out a special node $u$ and divide the rest into $2L = \Theta(\frac{n}{k}) = \Theta(n)$ (since $k=2$) bricks, from which we construct two roads of length $L$ each. See Figure~\ref{figure_coinflip_phases}:
\begin{enumerate}
    \item In odd phases we regard $u$ and the two roads as a $2$-road. That is, we add the two edges that connect $u$ to the roads.
    \item In even phases we regard the two roads as a single long road. That is, we add an extra edge that connects the last brick of the first road and the first brick of the second road.
\end{enumerate}

Every phase has $m$ rounds, and in each round we request all the edges of the structure of that phase (the $2$-road or single long road), in some arbitrary order.

By definition, in order to fully transition from one structure to the other, at least one road (of length $L$) must be recolored, because in odd phases they should disagree in color, and in even phases they should agree in color. As an upper bound for $OPT$, consider an algorithm that simply recolors one of the roads when a new phase begins, then $OPT$ pays $O(L)$ per phase. Next, consider $ALG$ on an odd-then-even pair of phases. In general, we say that $ALG$ is stable in a phase if it caches all the edges of that phase simultaneously. When $ALG$ is unstable, at least one edge is missing from its cache. There are two cases:
\begin{enumerate}
    \item $ALG$ is unstable by the end of the odd phase: If at some point during that phase $ALG$ was stable, it remains stable. Therefore we conclude that throughout the whole phase it was unstable, and paid at least $1$ per round. In total, $ALG$ pays for the two phases at least $m$ (due to the odd phase).

    \item $ALG$ is stable by the end of the odd phase: then the two roads have different colors and $ALG$ is unstable with respect to the new even phase. $ALG$ pays at least $1$ per round until either it stabilizes or the phase ends. Observe that since $k=2$, a brick is simply a path of three edges, and that the graph is simply a path of $8L-1$ edges whose red/blue ``parity'' is mismatched in the middle. We can order the path as a line, and correspond the edges to integers. When the missing edge, the \emph{hole}, is requested, choosing its color is equivalent to randomly moving the hole left or right. So $ALG$ stabilizes in the even phase if and only if the corresponding random walk reaches $4L-1$ or $-(4L-1)$ (starting from $0$, the middle), see Figure~\ref{figure_red_blue_chains_mismatch}. It is known\footnote{e.g., https://math.stackexchange.com/questions/288298/symmetric-random-walk-with-bounds} that the expected number of steps in such a random walk to reach either $a$ or $-b$ ($a,b>0$) is $a \cdot b$. Therefore if $m$ was infinite, $ALG$ would pay in expectation $\Omega(L^2)$. Because $m$ is finite, the expectation is smaller. However, Lemma~\ref{lemma_math_truncated_random_walk} proves that by choosing $m$ to be large enough, say $N$,  the truncation loses at most a constant factor, and still yields an expected cost of $\Omega(L^2)$.
\end{enumerate}

By choosing $m = \max\{N,L^2\}$ we get that $ALG$ pays in expectation $\Omega(L^2)$ for the pair of phases in either case. A repetition of enough phases makes the initialization cost (pre-phases) negligible, and yields a competitive ratio $\Omega(L) = \Omega(\frac{n}{k}) = \Omega(n)$ (recall that $k=2$).
\end{proof}

\begin{figure}[t]
	\centering
    \includegraphics[width=0.6\textwidth]{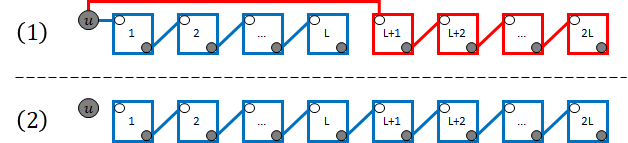}
    
	\caption{\small{The two states of graph in Lemma~\ref{lemma_RC_bad_competitive}. (1) During odd phases, requested edges imply a $2$-road. (2) During even phases, requested edges imply a single long road.}}
	\label{figure_coinflip_phases}
\end{figure}

\begin{figure}[t]
	\centering
    \includegraphics[width=0.45\textwidth]{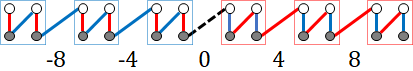}
    
	\caption{\small{The initial state of an unstable even phase induces a random walk on the integer line. The left (blue) road and the right (red) road form a path of edges with a ``hole'' marked by the dashed black edge whose location changes according to a random walk. For clarity we only marked some of the integers, those that correspond to edges between bricks.}}
	\label{figure_red_blue_chains_mismatch}
\end{figure}

It is plausible that Lemma~\ref{lemma_RC_bad_competitive} should generalize to a lower bound of $\Omega(\frac{n}{k})$ for $k \ge 2$. The proof should go almost exactly the same, except that now the random walk argument is not trivial. When $k>2$, a road is not simply a path, and there are two issues: (1) When randomly coloring an added edge, it is possible that two neighbouring edges would be removed, creating multiple holes that may help or interfere with the movement of other holes; (2) The random walk itself is not as clean as a walk on the integer line. There are still ``bottlenecks'' thanks to the edges between consecutive bricks, but one should be careful when analyzing the probabilities to advance a hole towards the ends of the roads. Moreover, even if everything is symmetric at the beginning of the even phase, after $ALG$ starts recoloring edges, it gets messier because different bricks may have different coloring to their ``inside edges'', affecting the probability of a hole crossing over each brick.

\begin{lemma}
\label{lemma_math_truncated_random_walk}
Consider a random walk that starts at $0$, moving at each step to either side with probability $\frac{1}{2}$ each. Let $E_m$ be the expected number of steps until stopping, either due to reaching $a > 0$ or $-b < 0$, or after $m$ steps. Then there is $N > 0$ such that for $m \ge N$, $E_m \ge \frac{1}{2} E_\infty$.
\end{lemma}

\begin{proof}
Stopping the walk due to capping the number of steps only decreases the expectation, therefore $E_1 \le E_2 \le \ldots \le E_\infty$. Let $p_i$ and $q_i$ be the probabilities that a non-truncated walk ends after exactly $i$ steps, and after more than $i$ steps, respectively ($q_i = \sum_{j=i+1}^{\infty}{p_j} > p_{i+1}$). Note that if $a+b$ consecutive steps all walk to the right, then we stop (must hit $a$), so for every block of $a+b$ consecutive steps the probability to end the walk is at least $\frac{1}{2^{a+b}}$. Therefore $i$ steps can be divided into disjoint blocks of consecutive steps, and we conclude that $p_{i+1} < q_i \le (1-\frac{1}{2^{a+b}})^{\floor{i/(a+b)}}$. To simplify, denote $z = (1-\frac{1}{2^{a+b}})$ and remember that $z < 1$. The exponential decrease of $q_i$ also implies that $E_\infty$ is bounded and well-defined, even without knowing its exact value ($E_\infty = a \cdot b$). Now:
$$
E_\infty - E_N
= \sum_{i=1}^{\infty}{i \cdot p_i} - \sum_{i=1}^{\infty}{\min\{i,N\} \cdot p_i}
= \sum_{i=N+1}^{\infty}{(i-N) \cdot p_i}
= \sum_{i=N}^{\infty}{(i+1-N) \cdot p_{i+1}}
$$
Note that $i - (a+b) < (a+b) \cdot \floor{i/(a+b)}$. Set $N = 1 + (c+1) \cdot (a+b)$ for integer $c \ge 1$ (to be determined), then $i+1-N < (a+b) \cdot (\floor{i/(a+b)} - c)$. We continue:
$$
< \sum_{i=N}^{\infty}{(i+1-N) \cdot z^{\floor{i/(a+b)}}}
< (a+b) \sum_{i=N}^{\infty}{(\floor{i/(a+b)} - c) \cdot z^{\floor{i/(a+b)}}}
$$
We can batch all the summands for  $j \cdot (a+b) \le i \le (j+1) \cdot (a+b)-1$ together. Continue by batching, and note that the first batch is incomplete (hence the strict inequality):
$$
< {(a+b)}^2 \sum_{j=c+1}^{\infty}{(j-c) \cdot z^j}
= (a+b)^2 z^c \sum_{j=1}^{\infty}{j \cdot z^j}
= (a+b)^2 z^c \cdot \frac{z}{(1-z)^2}
$$
Since $\lim_{c \to \infty}{(E_\infty - E_N)} = 0$, choosing a large enough $c$ yields $N$ such that $E_m \ge \frac{1}{2} E_\infty$ for any $m \ge N$.
\end{proof}


\textbf{Linear-formulation and rounding schemes:}
We can formalize the Caching in Matchings problem as a set of linear constraints, some of them arrive online and correspond to requests. This enables us to use  additional tools and techniques for designing competitive algorithms~\cite{Buchbinder_book}, in particular for caching~\cite{Generalized_Caching,GeneralizedCaching_Buchbinder}. The LP-formulation of our problem for a sequence $\sigma$ is given below.\footnote{The formulation implicitly assumes that the graph is bipartite, as in our problem. To generalize the formulation to general (non-bipartite) graphs we should require for every subset of $2r+1$ nodes that the sum of all variables of fixed color $c$ adjacent to them is at most $r$, by Edmonds's theorem~\cite{Edmonds1965MaximumMA} regarding the convex hull of the integer matchings.}

\begin{algorithm}
\caption{Fractional Caching in Matchings (Linear Program)}
    $\min
        \sum_{t=1}^{|\sigma|}{
            \sum_{c=1}^{k}{
                \sum_{e \in E}{
                    \Delta^t_{e,c}
                }
            }
        }
    $
    
    Such that:
    \begin{enumerate}[nosep]
        \item $\forall t: \sum_{c=1}^{k}{x^t_{\sigma(t),c}} \ge 1$ // Covering requested edges 
        \item $\forall t,v,c: \sum_{e \in E(v)}{x^t_{e,c}} \le 1$ // Packing proper coloring, $E(v)$ are the edges of $v$
        \item $\forall t,e,c: x^{t}_{e,c}-x^{t-1}_{e,c} \le \Delta^t_{e,c}$ // Technical for the formulation
        \item $\forall t,e,c: \Delta^t_{e,c}, x^t_{e,c} \ge 0$ // Non-negativity
    \end{enumerate}
    \label{problem_fractional_caching_in_matchings}
\end{algorithm}
    
In simple words, the LP-formulation of the problem consists of $\Theta(n^2 k)$ variables of the form $x_{e,c}$, one for every possible color of every possible edge. Technically, for defining the objective function of the LP, we define a different copy of each variable for each time, and also add the variables $\Delta^t_{e,c}$ to represent the growth of the variable $x_{e,c}$ from time $t-1$ to time $t$. Note that the formulation only considers growth, which is fine because growth corresponds to fetching into the cache (i.e., if $x^t_{e,c} - x^{t-1}_{e,c} < 0$ then the solution can have $\Delta^t_{c,e} = 0$). We denote by $\sigma(t)$ the request at time $t$, the covering constraints ensure that at any time the current request is fully cached, and the packing constraints correspond to (fractional) proper coloring. We refer to this formulation as \emph{fractional Caching in Matchings}. This is a special case of the online convex body-chasing problem, in the $L_1$ metric. Of course, the integer formulation requires changing the non-negativity constraints to $\forall t,e,c: x^t_{e,c} \in \{0,1\}$. While we can also add the constraints that $\forall t,e,c: x^t_{e,c} \le 1$, these are implied, and a reasonable algorithm should never increase a variable to more than $1$. Indeed, increasing variables has a cost, a covering constraint that is satisfied by $x^t_{e,c} > 1$ remains satisfied if $x^t_{e,c} = 1$, and packing constraints loosen and also remain satisfied for smaller values. To simplify notation, henceforth we omit the time superscript from $x_{e,c}$.

In terms of body-chasing, the dimension of our problem is $d = \Theta (k \cdot n^2)$ because there are $k$ colors and $\Theta(n^2)$ edges. While an $\Omega(\sqrt{d})$ lower bound on the competitive ratio of randomized algorithms is known for general convex body-chasing~\cite[see Lemma 5.4]{ConvexBodyChasingLB2020}, that bound necessitates an adversary that is much stronger than the adversary in our case. In our case, the adversary can only introduce a single temporary covering constraint at a time, while the rest of the constraints are fixed. Therefore, this lower bound \emph{cannot} be used to argue for an $\Omega(\sqrt{k} \cdot n)$ randomized lower bound for Caching in Matchings.

If the packing constraints are relaxed such that $\sum_{e \in E(v)}{x_{e,c}} \le 1 + \epsilon$, we can use a more general result on convex body-chasing in $L_1$ of Bhattacharya~\etal~\cite{L1BodyChasing2023} that applies to our case.

\begin{theorem}[Theorem 1.1 in \cite{L1BodyChasing2023}]
\label{theorem_convex_augmented}
For any $\epsilon \in (0,1]$, there is an $O(\frac{1}{\epsilon} \cdot \log \frac{k}{\epsilon})$ competitive algorithm for fractional Caching in Matchings where the bound for each packing constraint is $1+\epsilon$ instead of $1$.
\end{theorem}

Unfortunately, Theorem~\ref{theorem_convex_augmented} does not allow for $\epsilon = 0$. One can think of $\epsilon > 0$ as a kind  of resource augmentation, but it is  not  the natural augmentation for our problem. Rather than having extra colors, this augmentation allows having a bit extra from each existing color. It does not seem that the two types of augmentations, having more colors and having more of the same colors, are equivalent.

At a first glance, it looks like we can remove the need for some of the augmentation.

\begin{lemma}
\label{lemma_reduce_A_to_B}
Let $A$ be the algorithm claimed by Theorem~\ref{theorem_convex_augmented} for $\epsilon = \frac{1}{2k}$. It satisfies $\forall v: \sum_{c=1}^{k}{\sum_{e \in E(v)}{x_{e,c}}} \le k \cdot (1+\epsilon)$. We can transform it into an algorithm $B$ that guarantees:
\begin{enumerate}[nosep]
    \item $B$ satisfies the same constraints that $A$ does, the covering constraints, and the packing constraints up to $1+\epsilon$.
    \item For any sequence of requests $\sigma$, $cost(B(\sigma)) \le 2 \cdot cost(A(\sigma))$.
    \item For every node $v$ at any time, $\sum_{c=1}^{k}{\sum_{e \in E(v)}{x_{e,c}}} \le k$.
\end{enumerate}

\end{lemma}
\begin{proof}
The argument is applied independently to every edge $e$ so we omit it from the subscript and only use $c$ for the color. Denote the old coordinates (before change) of $A$ and $B$ by $x_c$ and $y_c$ respectively, and the new coordinates with prime, $x'_c$ and $y'_c$. Denote the function $(z)_+ \equiv \max \{ z,0 \}$. Define $a = \sum_c x_c$, $b = (2a - 1)_{+}$, and similarly $a'$ and $b'$. $B$ adjusts $y$ to maintain $0 \le y'_c \le x'_c$ and $\sum_c y'_c = b'$: First it reduces any $y_c > x'_c$ such that $y'_c = x'_c$. Then, as long as $\sum_c y'_c > b'$ it deletes arbitrary excess mass from the edge, and if $\sum_c y'_c < b'$ it adds mass in arbitrary coordinates with slack $x'_c > y'_c$. Note that we can always match the mass exactly, because $b' \le a'$ so it can always increase $y$ until $\forall c: y'_c = x'_c$. $b' \le a'$ holds because mass above total $1$ is never useful, so we may assume that $a' \le 1$, hence $b' \le a'$. Now we prove the claims in the order that we stated them:

\begin{enumerate}
    \item When edge $e_t$ is requested, $a'_{e_t} = 1$ therefore $b'_{e_t} = 1$ as well and the covering constraint is satisfied. In addition, since $y_{e,c} \le x_{e,c}$ for every $e$ and $c$, $B$ satisfies all the packing constraints (up to $\epsilon$) because $A$ does.

    \item Formally, denote $P = \sum_{c}{(x'_c - x_c)_{+}}$, $N = \sum_{c}{(x_c - x'_c)_{+}}$, and $I = \sum_{c}{(y_c - x'_c)_{+}}$. The initial invariant $y_c \le x_c$ gives $I \le N$. Then, clipping $y \to y' \le x'$ removes $I$ units, leaving a total mass of $b - I$, and the final adjustment increases by $(b' - (b-I))_{+}$. We divide into cases, and also use $a' - a = P - N$ (by definition):
    \begin{enumerate}
        \item If $b,b' > 0$: Then $b'-b = 2(a' - a)$ (no clipping), hence $((b'-b)+I)_{+} \le (2(P-N)+N)_{+} \le 2P$.

        \item If $b = 0$: Then $a \le \frac{1}{2}$, and $I=0$ (because $b - I \ge 0$), hence $(b'-b+I)_{+} = b' = (2a'-1)_{+} \le (2a' - 2a)_{+} \le 2P$.

        \item If $b' = 0$: Since $b - I \ge 0$, $(b' - (b-I))_{+} = (-(b-I))_{+} = 0 \le 2P$.
    \end{enumerate}
    We bounded the increase (cost) of $B$ which is $(b'-b+I)_{+}$ by twice the increase of $A$ ($P$), thus the factor-two bound holds per update and per edge, even when the total mass is unchanged and colors exchange mass.

    \item Fix the node $v$ with edges $E(v)$. Let $X$ be the set of edges of $v$ for which $a_e \le \frac{1}{2}$ and let $Y = E(v) \setminus X$. Denote also the slack $s_a$ such that $(k + \frac{1}{2}) - \sum_{e \in X \cup Y}{a_e} = s_a$, for $s_a \ge 0$. We need to prove that $\sum_{e \in Y}{b_e} \le k$ (note that $\sum_{e \in X}{b_e} = 0$ by definition of $X$ and the second invariant). There are two cases:
    \begin{enumerate}
        \item If $\sum_{e \in X}{a_e} + s_a \ge \frac{1}{2}$, then we are done because $\sum_{e \in Y}{b_e} \le \sum_{e \in Y}{a_e} \le (k+\frac{1}{2}) - \frac{1}{2} = k$.

        \item If $\sum_{e \in X}{a_e} + s_a < \frac{1}{2}$, then $\sum_{e \in Y}{a_e} > k$. Because each summand is bounded by $1$, $|Y| \ge k+1$. Then:
        $\sum_{e \in Y}{b_e} =
        2 \sum_{e \in Y}{a_e} - |Y| \le         2(k+\frac{1}{2}) - (k+1) \le k$.
         \qedhere
    \end{enumerate}
\end{enumerate}
\end{proof}

Since we chose $\epsilon = \frac{1}{2k}$ in Lemma~\ref{lemma_reduce_A_to_B}, by Theorem~\ref{theorem_convex_augmented} both $A$ and $B$ are $O(k \log k)$ competitive. While $B$ is not using more than a total of $k$ colors per node, it is not free of resource augmentation. We did not address the packing constraints specifically, but rather globally per node, and therefore the state of $B$ still allows the portion of some color in some node to exceed $1$ up to $1 + \frac{1}{2k}$.

Other than the issue of the non-natural resource augmentation, there is of course the question of finding a rounding scheme. In order to get a randomized solution, one must find a rounding scheme for the fractional solution. It would be nice if the solution could be done in two steps, first resolving the augmentation of either algorithm $A$ or $B$, and then finding a rounding scheme, but the previous paragraphs hint that perhaps such separation is not so simple if it even exists.

\subsection{Miscellaneous Proofs}
\label{appendix_subsection_misc_proofs}

We begin by re-stating and proving the theorem that derives a connection caching algorithm from a caching algorithm, by losing only a factor of $2$ in the competitive ratio. The explicit derivation is given in Algorithm~\ref{alg_compose_connection_caching}. For completeness, we formally define the caching problem.

\begin{problem}[Caching]
\label{problem_caching}
\looseness=-1
A sequence of page requests $p_1,p_2,\ldots$ is revealed one at a time. An algorithm maintains a cache of size $k$. When a page that is not in the cache is requested, the algorithm must bring it to the cache, possibly evicting another, and pays $1$. The algorithm may also bring additional pages into the cache, and the cost of the algorithm is $1$ per added page.
\end{problem}

\begin{algorithm}[ht]
    \SetAlgoLined
    \DontPrintSemicolon
    \KwIn{
         A sequence of connection requests $\sigma$, revealed one at a time.
    }

    
    \SetKwFunction{funcInit}{Init}
    \SetKwFunction{funcServe}{Serve}
    \SetKwFunction{funcMakeRoom}{MakeRoom}
    \SetKwProg{Fn}{Function}{:}{}
    
    \Fn{\funcInit{caching algorithm $A$}}{
        Set the cache $M = \emptyset$. \;
        $\forall u \in V$: Initialize an instance $A_u$ and denote its virtual cache by $M_{u}$.
    }
    
    \Fn{\funcServe{request $(u,v)$}}{
        Regardless of whether $(u,v) \in M$, feed $(u,v)$ to $A_u$ and to $A_v$.\;
        If $A_u$ added 
        $(u,b)$ (for any $b \in V$) and 
        $(u,b) \in M_{b}$ then 
        add $(u,b)$ to $M$.
        If $A_u$ evicted
        $(u,b)$ and $(u,b) \in M$, then evict it from $M$.\;
        Similarly, if $A_v$ added 
        $(b,v)$ and $(b,v) \in M_{b}$ then 
        add $(b,v)$ to $M$. 
        If $A_v$ evicted
        $(b,v)$ and $(b,v) \in M$ then evict it from $M$.\;
    }
    
    \caption{Decoupling Connection Caching}
    \label{alg_compose_connection_caching}
\end{algorithm}

\theoremDecouplingCCtoCaching*

\begin{proof}
The idea of the proof is to charge the cost of the connection caching algorithm to the virtual local caching algorithms that run on each node, while maintaining the \emph{invariant} that an edge (connection) $(a,b)$ is cached if and only if it is cached in both virtual caches of $a$ and $b$.

To simplify notation, denote $\rho = c(r,k)$, so that $A$ is a $\rho$-competitive caching algorithm with additive constant $\delta \ge 0$.\footnote{While $\rho$ is a function of $k$ and $r$, the exact dependence does not matter for the analysis.}
Fix the sequence of requests $\sigma$ for connections. For a vertex $u \in V$, let $\sigma_{u} \subseteq \sigma$ consist of the requests for edges $(u,w)$ for some $w \in V$. Also, denote by $\costFull{ALG(\sigma)}{\sigma_{u}}$ the cost of a connection caching algorithm $ALG$ over the sequence $\sigma$ when we only consider costs due to requests in $\sigma_{u}$. Let $OPT$ be the optimal algorithm for serving $\sigma$. The cost of each request $(u,v) \in \sigma$ is counted both in $\costFull{OPT(\sigma)}{\sigma_{u}}$ and $\costFull{OPT(\sigma)}{\sigma_{v}}$, thus $\sum_{u \in V}{\costFull{OPT(\sigma)}{\sigma_u}} = 2 \cdot \cost{OPT(\sigma)}$.

We define $B_{u}$ to be an offline connection caching algorithm that minimizes $\costFull{ALG(\sigma)}{\sigma_{u}}$ over all algorithms $ALG$, therefore $\costFull{B_{u}(\sigma)}{\sigma_{u}} \le \costFull{OPT(\sigma)}{\sigma_{u}}$. By our definitions we can think of $B_{u}$ as a caching algorithm at $u$ running over $\sigma_{u}$. To clarify, since $k \ge 2$, $B_u$ can reserve the first cache slot of any $v (\ne u) \in V$ for the connection $(v,u)$ while using the rest of the slots to serve requests of edges $(v,w)$ for $w \ne u$.

We charge the cost of the algorithm for connection caching, denote it by $A'$, on every request $(u,v)$ to the virtual cost of either $A_u$ or $A_v$, depending on the state of the cache when $(u,v)$ arrives as follows. By \emph{virtual cost} we refer to the cost that each of these instances ``thinks'' that it pays, when $A'$ feeds $(u,v)$ to it in order to use its decisions.
Concretely, if $A'$ fetched an edge $(u,b)$ for some $v \ne b \in V$, then by the invariant both $A_u$ and $A_b$ have it cached, but previously $A_u$ did not have it cached (the state of $A_b$ did not change). So we can charge the cost of $A'$ to the virtual cost of $A_u$. Similarly we can charge the fetching of an edge $(b,v)$ for $u \ne b \in V$ to the virtual cost of $A_v$. Finally, if $A'$ fetched $(u,v)$, either $A_u$ or $A_v$ (or both) must have fetched it too, so we can charge that one. Overall, $\cost{A'(\sigma)} \le \sum_{u \in V}{\cost{A_u(\sigma_u)}}$. Since $A_u$ is $\rho$-competitive with additive term $\delta$ against $B_{u}$ when $B_{u}$ is viewed as a caching algorithm, we can wrap up as follows.
$$\cost{A'(\sigma)}
\le \sum_{u \in V}{\cost{A_u(\sigma_u)}}
\le \sum_{u \in V}{(\rho \cdot \costFull{B_{u}(\sigma)}{\sigma_{u}} + \delta)}
$$
$$
\le \rho \sum_{u \in V}{\costFull{OPT(\sigma)}{\sigma_{u}}} + |V| \cdot \delta
= 2\rho \cdot \cost{OPT(\sigma)} + |V| \cdot \delta$$
Note that the $2$ in the multiplicative term originates from the two sides of a connection, and the additive term is due to each node contributing separately up to $\delta$. If $A$ is randomized, the same analysis can be repeated as follows:
$$\expect{\cost{A'(\sigma)}}
\le \mathbb{E} \Big [ \sum_{u \in V}{\cost{A_u(\sigma_u)}} \Big ]
= \sum_{u \in V}{\expect{\cost{A_u(\sigma_u)}}}
\le \sum_{u \in V}{(\rho \cdot \costFull{B_{u}(\sigma)}{\sigma_{u}} + \delta)}
$$
$$
\le \rho \sum_{u \in V}{\costFull{OPT(\sigma)}{\sigma_{u}}} + |V| \cdot \delta
= 2\rho \cdot \cost{OPT(\sigma)} + |V| \cdot \delta$$
\end{proof}

\begin{remark}
\label{remark_general_decoupling}
Theorem~\ref{theorem_decoupling_competitive} and Algorithm~\ref{alg_compose_connection_caching} can be generalized to other caching variants, such as \emph{generalized caching} (edges may have different sizes in cache, and different costs) and \emph{caching with rejections}~\cite{RejectionCaching}. The argument remains the same: due to the invariant, the cost paid for a connection can be charged to the cost of (at least) one local-cache cost of the connection's end-points.
\end{remark}



\begin{lemma}
\label{lemma_even_graph_factorization}
The edges of a clique of $n$ nodes can be divided into $n-1$ disjoint matchings if and only if $n$ is even.
\end{lemma}

\begin{proof}
A complete graph with $n=2\ell+1 \ge 3$ nodes (odd) has $n \cdot \ell$ edges, and each matching can contain at most $\ell$ edges, so $n-1$ colors are insufficient.
A \emph{$1$-factorization} of a graph divides its edges into disjoint perfect matchings, and it is known how to find such a factorization for any complete graph with an even $n$. Constructively: pick $(n-1)$ nodes to form a regular polygon in the plane, and place the remaining node in the center. Each of the $n-1$ perfect matchings is defined by matching the center to one of the nodes, and matching each remaining node to its ``reflection'' with respect to the line spanned by the center's matched edge (the center's edge is perpendicular to all other edges of the matching).
\end{proof}

\begin{lemma}
\label{lemma_k_upper_bound_general}
Consider Caching in Matchings on a general graph. For fixed $k \ge 2$, there is an $O(n^2 k^2)$ competitive deterministic algorithm.
\end{lemma}

\begin{proof}
We essentially repeat the standard phases-based proof that is used to show that a standard caching algorithm with a cache of size $r$ is $r$-competitive. The difference is that our effective cache size is $r = \frac{nk}{2}$, and that in each request we may completely reorganize the cache at a cost of $O(r)$ instead of only paying $1$ to cache the new request.

We define an online algorithm that acts in phases as follows. Every phase begins with an empty cache. During a phase, when a request arrives, check if there is \emph{any} coloring of the edges that fits the existing edges and the new one. If there is, the phase continues, change the matchings accordingly to fit all the edges together. If not, we evict all the edges from the matchings and start a new phase. There can be at most $r = \frac{nk}{2}$ edges in the matchings overall, so a phase contains at most $r$ unique edges. When the matchings have $i \le r$ edges in total, the recoloring cost is at most $i$ plus $1$ to insert a new edge. Note that we do not always have to recolor all the edges, in particular the first $k$ edges always have room without recoloring. We get that, per phase: $cost(ALG) \le r + \sum_{i=k+1}^{r}{i} < r + \frac{r(r+1)}{2} = \frac{r(r+3)}{2}$.

Next we argue that $OPT$ pays at least $1$ per phase, which yields the competitive ratio $\frac{cost(ALG)}{cost(OPT)} = O(n^2 k^2)$. Denote the subsequence of requests of phase $i$ by $\phi_i$ and its first request by $r_i$. Then $r_1$ incurs a cost because the cache is initially empty. For $i>1$, we claim that either $r_i$ or some request in $\phi_{i-1} \setminus \{ r_{i-1} \}$ incurs a cost. To see why, denote by $q$ the edge requested by $r_i$. Note that $q$ was not requested during $\phi_{i-1}$ since by definition it started a new phase. Moreover, there is no configuration of the cache in which $q$ could be part of $\phi_{i-1}$ or else it would have been, because $ALG$ fully rearranges its cache to pack additional edges, if possible. So if $r_i$ does not incur a cost, then $q$ must have been cached by $OPT$ during all of $\phi_{i-1}$, particularly right after $r_{i-1}$ was served, and there must be some other edge $q'$ requested during $\phi_{i-1} \setminus \{ r_{i-1} \}$ that was not cached when requested.
\end{proof}

\bibliography{reference}

\end{document}